\documentclass[aps,prd,twocolumn,superscriptaddress]{revtex4-2}

\usepackage{makecell} 
\usepackage{microtype}
\usepackage{graphicx}
\usepackage{dcolumn}
\usepackage{bm}
\usepackage{caption} 
\usepackage{tensor} 
\usepackage[hidelinks]{hyperref} 
\usepackage{amssymb}
\usepackage{subcaption}
\usepackage{soul}
\usepackage{stmaryrd}
\usepackage{lmodern}
\usepackage{amsmath,amssymb,amsfonts,mathtools,bm}
\usepackage{mathrsfs}
\usepackage{amsthm}
\usepackage{booktabs,array}
\usepackage{enumitem}
\usepackage{xcolor}
\usepackage[nameinlink,noabbrev]{cleveref}
\usepackage{listings}

\usepackage[toc,page]{appendix}

\newtheorem{theorem}{Theorem}[section]
\newtheorem{proposition}[theorem]{Proposition}
\newtheorem{lemma}[theorem]{Lemma}
\newtheorem{corollary}[theorem]{Corollary}
\theoremstyle{definition}

\theoremstyle{remark}
\newtheorem{remark}[theorem]{Remark}

\newcommand{\dd}{\mathrm d}
\newcommand{\ii}{\mathrm i}
\newcommand{\e}{\mathrm e}
\newcommand{\Om}{\Omega}
\newcommand{\Sig}{\Sigma}
\newcommand{\PP}{\mathcal P}
\newcommand{\QQ}{\mathcal Q}
\newcommand{\RR}{\mathcal R}
\newcommand{\Base}{\mathcal B}
\newcommand{\Target}{\mathcal T}
\newcommand{\Ham}{\mathscr H}

\newcommand{\Rea}{\operatorname{Re}}

\newcommand{\Tr}{\operatorname{Tr}}
\newcommand{\bars}[1]{\overline{#1}}
\newcommand{\epss}{\epsilon_0}
\newcommand{\alphao}{\alpha_0}
\newcommand{\acc}{\mathfrak a}

\newcommand{\Bal}{\mathfrak B_{+}}

\begin{document}

\title{Potential-space geometry and scalar rigidity of the Plebański-Demiański ansatz without a cosmological constant}

\author{Leonel Bixano}
    \email{Contact author: leonel.delacruz@cinvestav.mx}
\author{Tonatiuh Matos}%
 \email{Contact author: tonatiuh.matos@cinvestav.mx}
\affiliation{Departamento de F\'{\i}sica, Centro de Investigaci\'on y de Estudios Avanzados del Instituto Politécnico Nacional, Av. Instituto Politécnico Nacional 2508, San Pedro Zacatenco, M\'exico 07360, CDMX.
}%

\date{\today}

\begin{abstract}
    We analyze stationary, axisymmetric Einstein–Maxwell–dilaton theory with zero cosmological constant in the conformal Carter class, employing a generalized Ernst approach on the space of potentials. We begin by establishing a complete local correspondence between Plebański–Demiański principal coordinates and Weyl data, covering the electric, magnetic, and twist potentials as well as the 1-forms \(A\), \(B\), and \(C\). In the source-free Weyl reduction, requiring the invariant area density to be harmonic singles out the \(\Lambda=0\) sector, and the principal chiral splitting isolates both electromagnetic misalignment and the part due to acceleration. Next, we obtain the Einstein reconstruction equations for general local EMD configurations in this setting. In a physical Carter domain with an ordinary scalar, the bilinear factor \(\Om=1-\acc pq\) enforces \(C=0\) and \(B_-=0\). With fixed \(\alphao\neq0\), the scalar field equation further rules out the charged PD Coulomb-type 1-form. Lastly, assuming an exact quadratic Stäckel scalar 1-form, the conditions of exactness together with reconstruction restrict the system to a single aligned branch, but consistency with the full potential equations removes this branch for arbitrary smooth Carter structure functions. Consequently, within the quadratic Stäckel class studied here, there is no local nonconstant solution with an ordinary dilaton. This conclusion does not extend to phantom scalars, conformal factors of higher degree, or non-Stäckel 1-forms.
\end{abstract}

\maketitle

\section{Introduction}

The Plebański–Demiański (PD) family is the canonical type-D solution of the Einstein–Maxwell equations with a non-null electromagnetic field aligned with the repeated principal null directions of the Weyl tensor \cite{Plebanski:1976gy,Griffiths:2005qp}. Its separability properties, quartic structure functions, acceleration parameter, and conformal Killing–Yano structure are well established \cite{Carter:1968ks,Kubiznak:2007kh}. More recently, Ovcharenko and Podolský constructed a broad rotating type-D electrovacuum class with a fully non-aligned electromagnetic field \cite{Ovcharenko:2025cpm}. Its local conformal-to-Carter characterization has since been sharpened \cite{Ovcharenko:2026pow,Nakajima:2026lkr}, and a new PD-adapted parametrization has clarified the full parameter space and its aligned limit \cite{Furugori:2026qbb}. Both the aligned and non-aligned sectors possess a non-degenerate conformal Killing–Yano 2-form, although in the non-aligned case it does not generate the full symmetry tower \cite{Gray:2025lwy}.

Weyl-coordinate descriptions are already known for important accelerating type-D subfamilies, notably the spinning \(C\)-metric \cite{Bicak:1999sa}, and the dyonic PD seed has been treated through the standard Ernst formalism \cite{Astorino:2023elf}. We therefore claim neither the use of Weyl variables nor the existence of an Ernst representation as new. Our contribution is the closed local reconstruction of the five real potentials and the \((A,B,C)\) target-space coframe, together with the Weyl integrability conditions, Hamiltonian reconstruction, principal-chirality interpretation, and a sequence of local rigidity results culminating in a full-system no-go theorem for the quadratic Stäckel scalar class.

The potential-based formulation of stationary and axisymmetric gravitational configurations originates from Ernst’s groundbreaking work, in which he reformulated the vacuum and Einstein–Maxwell field equations using complex potentials \cite{Ernst:1967wx,Ernst:1967by}. This framework was later endowed with a geometric interpretation via the theory of harmonic maps. Specifically, Matos and Plebański demonstrated that the stationary and axisymmetric vacuum Einstein equations can be cast as a harmonic map from the reduced space-time into a suitable potential space \cite{Matos:1994hm}. The harmonic-map methodology was further extended to five-dimensional stationary and axisymmetric gravity, where the corresponding reduced equations were analyzed through the geometry of the associated target manifold \cite{Matos:1994qv}. Matos et al subsequently generalized this program to the Einstein–Maxwell–dilaton and Einstein–Maxwell–phantom setting \cite{Matos:2000za,Matos:2010pcd,Matos:2000ai}.

The formulation employed in the present article is the generalized Ernst-like framework developed in Refs.~\cite{Bixano:2026xum,Bixano:2026ouq,Bixano:2026duf}. In Ref.~\cite{Bixano:2026xum}, the stationary axisymmetric Einstein-Maxwell-Scalar gield system was reformulated in terms of the five real potentials \(Y^{A}=(f,\epsilon,\psi,\chi,\kappa)\), the corresponding five-dimensional target-space metric was obtained for arbitrary scalar field coupling, and the field equations were written in terms of the three differential 1-forms \(A\), \(B\), and \(C\). This construction was extended in Ref.~\cite{Bixano:2026ouq} to broader Einstein-ModMax-Scalar theories. The geometric and algebraic structure of the potential space was subsequently analyzed in Ref.~\cite{Bixano:2026duf}, where its visible, hidden, sectorial, and discrete symmetries were classified, and the harmonic one-potential reduction was identified with an affinely parametrized geodesic flow governed by a finite-dimensional Hamiltonian system. Section~\ref{sec:potential} summarizes the elements of this framework required below and fixes our conventions before specializing to the Plebański-Demiański geometry. The original contribution of the present work begins with the complete translation of the Plebański-Demiański family into this potential-space language and uses the same \(A,B,C\) field equations and Hamiltonian tensor to investigate its possible dilatonic extensions and the associated rigidity obstructions.

%

The results presented here supplement, but do not duplicate, existing no-go theorems. Kocherlakota and Narayan exclude a minimally coupled massless real scalar in a different asymptotically flat, doubly separable metric class \cite{Kocherlakota:2025cwq}. Conversely, stationary axisymmetric Einstein–Maxwell–dilaton–axion solutions with a Killing tensor exist in the broader Benenti–Francaviglia class, which is generically of Petrov type I \cite{Galtsov:2025nia}. Our bilinear theorem assumes neither a polynomial form for the Carter functions nor electromagnetic alignment. The stronger full-system no-go result instead concerns the explicitly stated polynomial conformal factor and exact quadratic Stäckel scalar 1-form in a physical Carter region.

The article is organized into three numbered sections. Section~\ref{sec:potential} presents the Ernst-like potential theory and fixes all conventions used later. Section~\ref{sec:PDdictionary} constructs the full PD dictionary and then interprets type D, Maxwell alignment, acceleration, conical data, and hidden symmetry directly in potential space. Section~\ref{sec:rigidityMain} develops the Carter reconstruction identities together with the conditional no-go and rigidity theorems. More extensive computations in base calculus, Routh reduction, the Newman-Penrose formalism, and coefficient verification are deferred to the appendices, so that the main exposition remains readable without compromising reproducibility.

\section{The Ernst-like potential theory}
\label{sec:potential}

This section summarizes the part of the framework developed in Refs.~\cite{Bixano:2026uqk,Bixano:2026ouq} that will be needed in the remainder of the article.  Consider a smooth four-dimensional spacetime $\mathcal M$ equipped with a Lorentzian metric $g$, a real dilatonic scalar field $\phi$, and an electromagnetic four-potential $\mathcal A$ with associated field strength $F=\dd\mathcal A$.  We adopt the signature $(-,+,+,+)$ and work in the Einstein frame, with action
\begin{equation}\label{eq:action}
     S=\int \dd^4x\sqrt{-g}\left[
     R-2\epss(\nabla\phi)^2-\e^{-2\alphao\phi}F_{\mu\nu}F^{\mu\nu}
     \right],
\end{equation}
where $R$ is the scalar curvature of $g$, $(\nabla\phi)^2=g^{\mu\nu}\partial_\mu\phi\partial_\nu\phi$, $\alphao$ is the fixed dilatonic coupling, and $\epss=+1$ denotes the ordinary dilatonic scalar field while $\epss=-1$ denotes its phantom continuation. We set
\begin{equation} \label{eq:kappaC}
     \kappa:=\e^{-\alphao\phi},
     \qquad
     C:=-\dd\ln\kappa=\alphao\,\dd\phi.
\end{equation}
There is no cosmological constant and no scalar potential.

\subsection{Spacetime, orbit space, and real potentials}

Three geometries enter the reduction and will never be identified with one another:
{\small
\begin{equation*}
     (\mathcal M,g,F,\phi)\longrightarrow
     (\Base,\gamma;\rho)\overset{Y}{\longrightarrow}(\Target,G),
     \qquad \gamma=\dd\rho^2+\dd z^2.
\end{equation*}
}
Here $\Base:=\mathcal M/(\mathbb R\times U(1))$ is the two-dimensional orbit space of the stationary and axial isometries on a regular circular patch, called the \emph{Weyl base}. The tensor $\gamma$ is the flat metric in its conformal class in Weyl coordinates $(\rho,z)$, where $\rho:=\sqrt{-\det(g_{ab})}$, $a,b\in\{t,\varphi\}$, is the invariant area density of the Killing orbits. The map $Y$ collects the five real potentials and maps $\Base$ into the five-dimensional target manifold $\Target$, with target metric $G$ defined below. Thus $\Base$ is two-dimensional, while $\Target$ is a five-dimensional potential space.

The spacetime Hodge operator is denoted by ${}^{\star}$ and is fixed by the orientation $\eta_{t\rho z\varphi}=+\sqrt{-g}$ through $({}^{\star}F)_{\mu\nu}=\tfrac12\eta_{\mu\nu}{}^{\alpha\beta}F_{\alpha\beta}$, with signature $(-,+,+,+)$ it obeys ${}^{\star}{}^{\star}F=-F$ on two-forms.  The symbol $\star_{\Base}$ acts only on 1-forms of the Weyl base and is fixed by $\star_{\Base}\dd\rho=\dd z$, $\star_{\Base}\dd z=-\dd\rho$ and the orientation $\dd\rho\wedge\dd z>0$.  The symbol $\wedge$ always denotes the exterior product on the manifold on which the displayed forms live.  The operator $\star_{\rm PD}$ introduced in Section~\ref{sec:PDdictionary} is the physical base Hodge operator transported to the principal chart $(p,q)$. The symbol $\star_0$, introduced still later, is an auxiliary principal complex structure.  Complex conjugation of $A$ and $B$ is unrelated to every Hodge operation.  Finally, the Weyl metric function $k$ is distinct from the structural constant $k_{\rm PD}$ in the PD quartics.

Let $\xi_{(t)}=\partial_t$ denote the stationary Killing vector and $\xi_{(\varphi)}=\partial_\varphi$ the axial Killing vector, where $\varphi$ is the axial coordinate. Under the circularity condition, the metric can be written in the Weyl–Lewis–Papapetrou form
{\small
\begin{equation}\label{eq:WeylMetric}
     \dd s^2=-f(\dd t-\omega\dd\varphi)^2
     +f^{-1}\left[\e^{2k}(\dd\rho^2+\dd z^2)+\rho^2\dd\varphi^2\right],
\end{equation}
}
here $f=-g(\xi_{(t)},\xi_{(t)})>0$ holds on the stationary region, $\omega$ denotes the rotational function that appears in the Killing sector, and $k$ is the conformal factor of the orbit metric.  Each of these three functions depends solely on $(\rho,z)$.  In a gauge that preserves the symmetries, the electromagnetic four-potential takes the form $\mathcal A=A_t(\rho,z)\dd t+A_\varphi(\rho,z)\dd\varphi$.

The real potentials $(f,\epsilon,\psi,\chi,\kappa)$ are defined by (See Ref.~\cite{Matos:1994hm,Matos:2000ai,Matos:2000za,Matos:2010pcd,Bixano:2026xum,Bixano:2026ouq,Bixano:2026duf})
\begin{equation} \label{eq:potdefinitions}
\begin{aligned}
    &\psi=2A_t,
     \quad
     \dd\chi=\kappa^2\star_{\Base}
     \left[\frac{2f}{\rho}(\omega\dd A_t+\dd A_\varphi)\right],
     \\
     &\Theta:=\dd\epsilon-\psi\dd\chi
     =\star_{\Base}\left(\frac{f^2}{\rho}\dd\omega\right).
\end{aligned}
\end{equation}
The scalar $\psi$ is the electric potential in the normalization of Refs.~\cite{Bixano:2026xum,Bixano:2026ouq}. $\chi$ is the magnetic potential obtained by dualizing the axial Maxwell equation. $\epsilon$ is the twist potential obtained by dualizing the mixed Einstein equation, and $\Theta$ is the resulting real twist 1-form. These potentials are local and are defined up to additive constants on a simply connected patch.  Their defining relation implies the identity $\dd\Theta=-\dd\psi\wedge\dd\chi$.

The potential map before mentioned is therefore
\[
Y:\Base\longrightarrow\Target,\qquad Y^A=(f,\epsilon,\psi,\chi,\kappa).
\]
The three 1-forms in the Matos-Bixano (MB) coframe are
\begin{equation} \label{eq:ABC}
\begin{aligned}
     &A=\frac{1}{2f}(\dd f-\ii\Theta),\\
     &B=-\frac{1}{2\sqrt f}\left(\kappa\dd\psi-\frac{\ii}{\kappa}\dd\chi\right),
     \\ &C=-\dd\ln\kappa .
\end{aligned}
\end{equation}
The basic dictionary, valid before any PD specialization, is
{\footnotesize
\begin{equation}
 \begin{gathered}
     \dd f=f(A+\bars A),\qquad
     \Theta=\ii f(A-\bars A),\qquad
     \dd\kappa=-\kappa C,\\
     \dd\psi=-\frac{\sqrt f}{\kappa}(B+\bars B),\quad
     \dd\chi=-\ii\sqrt f\,\kappa(B-\bars B),\quad
     \dd\epsilon=\Theta+\psi\dd\chi,\\
     \dd\omega=-\frac{\rho}{f^2}\star_{\Base}\Theta,\qquad
     \dd A_\varphi=-\frac\omega2\dd\psi
     -\frac{\rho}{2f\kappa^2}\star_{\Base}\dd\chi.
 \end{gathered}
 \label{eq:basicDictionary}
\end{equation}
}
Thus $(A,B,C)$ retain all five target differentials and reconstruct the two cyclic spacetime functions. They are the pullbacks of a complex coframe on the target space with metric
{\footnotesize
\begin{equation}
\begin{aligned}
     \dd s^2_{\Target}
     =\frac{\dd f^2+(\dd\epsilon-\psi\dd\chi)^2}{2f^2}
     -\frac{1}{2f}\left(\kappa^2\dd\psi^2+\kappa^{-2}\dd\chi^2\right) 
     +\frac{2\epss}{\alphao^2}\frac{(\dd\kappa)^2}{\kappa^2}.
\end{aligned}
 \label{eq:targetmetric}
\end{equation}
}
Hence, $G$ is the pseudo-Riemannian metric whose line element is given by \eqref{eq:targetmetric}, for $\epss=+1$ it has signature $(3,2)$ in the order indicated. 

\subsection{Why the 1-forms \texorpdfstring{$(A,B,C)$}{(A,B,C)} have this form}

The combinations in \eqref{eq:ABC} are a complex semi-null coframe adapted to the metric of the potential space derived in \cite{Bixano:2026xum}. To see this directly, use the symmetric product of 1-forms. Since $\Theta$ and $\dd f$ are real,
{\footnotesize
\begin{equation}\label{eq:ABCcoframeproof}
     2A\bars A=\frac{\dd f^2+\Theta^2}{2f^2},
     \quad
     2B\bars B=\frac{\kappa^2\dd\psi^2+\kappa^{-2}\dd\chi^2}{2f},
     \quad
     C^2=\frac{(\dd\kappa)^2}{\kappa^2}.
\end{equation}
}
Therefore
\begin{equation}
 \dd s^2_{\Target}=2A\bars A-2B\bars B
 +\frac{2\epss}{\alphao^2}C^2.
 \label{eq:targetABC}
\end{equation}
Consequently, $A$ is determined by the gravitational and twist block, $B$ by the electric and magnetic block with exactly the dilatonic weights $(\kappa,\kappa^{-1})$, and $C$ by the logarithmic scalar field. 

By definition,
\begin{equation}
     \dd\Theta=-\dd\psi\wedge\dd\chi.
     \label{eq:HeisenbergCurvature}
\end{equation}
This relation holds irrespective of the field equations and implies that the gravito-rotational potential is fibered over the electromagnetic plane. It is the differential source of the mixed term in $A$ and underlies the manifest Heisenberg-type target-space symmetries analyzed in \cite{Bixano:2026duf}.

\subsection{Compact field equations and two distinct constant-scalar limits}

For 1-forms on $\Base$, we write $X\cdot Y:=\gamma^{ij}X_iY_j$ for the complex-bilinear contraction, this pairing is not Hermitian, and complex conjugation is always indicated explicitly. For a scalar function $u$, $Du$ denotes its $\gamma$-gradient, and for $X=X_\rho\dd\rho+X_z\dd z$ we define the weighted divergence by $D(\rho X):=\partial_\rho(\rho X_\rho)+\partial_z(\rho X_z)$. The field equations obtained in Ref.~\cite{Bixano:2026xum} are
\begin{subequations}\label{eq:FieldEquations}
\begin{align}
     \frac1\rho D(\rho A)
     &=B\cdot\bars B+A\cdot(A-\bars A),
     \label{eq:Aeq}\\
     \frac1\rho D(\rho B)
     &=-\frac12B\cdot(A-3\bars A)+C\cdot\bars B,
     \label{eq:Beq}\\
     \frac1\rho D(\rho C)
     &=\frac{\alphao^2}{2\epss}\left(B\cdot B+\bars B\cdot\bars B\right).
     \label{eq:Ceq}
\end{align}
\end{subequations}
Equivalently, $D(\rho X)=\star_{\Base}\dd(\rho\star_{\Base}X)$. 

There are two separate types of limits that should not be confused. Pure Einstein–Maxwell theory is recovered by eliminating the dynamical dilaton: concretely, one sets $\kappa = 1$, $C = 0$, and decouples the scalar by fixing \(\alphao = 0\), keeping only the $A$ and $B$ equations. In contrast, for fixed $\alphao \neq 0$, the condition $C = 0$ corresponds to a particular field configuration within EMSF, and the $C$ equation still applies. It then enforces
\begin{equation}
 B\cdot B+\bars B\cdot\bars B=0,
 \quad\hbox{equivalently}\quad \Rea(B\cdot B)=0.
 \label{eq:constantDilatonConstraint}
\end{equation}

\subsection{Exterior form of the field equations}

For later substitutions in the $(p,q)$ chart, it is crucial to specify the two-dimensional operations in an unambiguous way. On the Weyl base we choose the orientation so that $\dd\rho\wedge\dd z$ is positive, and we define $\star_{\Base}\dd\rho=\dd z$, $\star_{\Base}\dd z=-\dd\rho$. If $X=X_\rho\dd\rho+X_z\dd z$ and $Y=Y_\rho\dd\rho+Y_z\dd z$ are complex 1-forms, then a straightforward computation of their wedge product yields
\begin{equation}
\begin{aligned}
     &\star_{\Base}(X\wedge\star_{\Base}Y)=X_\rho Y_\rho+X_zY_z,
     \\
     &\star_{\Base}\dd(\rho\star_{\Base}X)
     =\partial_\rho(\rho X_\rho)+\partial_z(\rho X_z).
\end{aligned}
 \label{eq:HodgeProductDivergence}
\end{equation}
Consequently \eqref{eq:Aeq}-\eqref{eq:Ceq} are exactly equivalent to
{\footnotesize
\begin{align}
     \frac1\rho\star_{\Base}\dd(\rho\star_{\Base}A)
     &=\star_{\Base}(B\wedge\star_{\Base}\bars B)
     +\star_{\Base}[A\wedge\star_{\Base}(A-\bars A)],
     \nonumber\\
     \frac1\rho\star_{\Base}\dd(\rho\star_{\Base}B)
     &=-\frac12\star_{\Base}[B\wedge\star_{\Base}(A-3\bars A)]
     +\star_{\Base}(C\wedge\star_{\Base}\bars B),
     \nonumber\\
     \frac1\rho\star_{\Base}\dd(\rho\star_{\Base}C)
     &=\frac{\alphao^2}{2\epss}\star_{\Base}
     \left(B\wedge\star_{\Base}B+\bars B\wedge\star_{\Base}\bars B\right).
     \label{eq:ABCexterior}
\end{align}
}
\subsection{Reduced action and Hamiltonian tensor}

The stationary-axisymmetric equations follow from the harmonic-map functional
{\small
\begin{equation}
     I_{\rm pot}=\int_{\Base}\rho\,
     G_{AB}(Y)\,D Y^A\cdot D Y^B\,\dd\rho\dd z,
     \quad
     Y^A=(f,\epsilon,\psi,\chi,\kappa),
     \label{eq:potentialAction}
\end{equation}
}
whose Euler-Lagrange equation is
\begin{equation}
     \frac1\rho D(\rho D Y^A)
     +\widehat\Gamma^A{}_{BC}(Y)\,D Y^B\cdot D Y^C=0.
     \label{eq:harmonicMapEquation}
\end{equation}
where $\widehat\Gamma^A{}_{BC}$ denote the Levi-Civita connection coefficients of the target metric $G$, the capital indices refer to the five target coordinates, and contractions on the base are taken with $\gamma$.  Expressing \eqref{eq:harmonicMapEquation} in the coframe \eqref{eq:ABC} yields precisely \eqref{eq:Aeq}-\eqref{eq:Ceq}.

The pullback of the potential-space metric is the \emph{Hamiltonian tensor}
{\small
\begin{equation}
     \Ham_{ij}:=G_{AB}(Y)\,\partial_iY^A\partial_jY^B
     =2\Rea(A_i\bars A_j-B_i\bars B_j)
     +\frac{2\epss}{\alphao^2}C_iC_j .
     \label{eq:HamiltonianTensor}
\end{equation}
}
Here $i,j\in\{\rho,z\}$ are Weyl-base indices. Following the Hamiltonian construction of Ref.~\cite{Bixano:2026duf}, we call $\Ham$ the \emph{Hamiltonian tensor}. The Weyl function $k$ is not an independent target coordinate, it is reconstructed from the trace-free part of \eqref{eq:HamiltonianTensor}.

In the Weyl chart, the remaining Einstein equations are precisely
\begin{equation}
     k_{,\rho}=\frac\rho2(\Ham_{\rho\rho}-\Ham_{zz}),
     \qquad
     k_{,z}=\rho\Ham_{\rho z}.
     \label{eq:kHamiltonianConstraints}
\end{equation}
If \(\Ham^\circ := \Ham-\frac12\bigl(\Tr_\gamma\Ham\bigr)\gamma\) denotes the trace-free part of the Hamiltonian tensor with respect to \(\gamma=\dd\rho^2+\dd z^2\), then, in the Weyl coframe \((\dd\rho,\dd z)\),
\begin{equation}
    \begin{split}
    \Ham^\circ
    ={}&
    \frac12
    \left(
    \Ham_{\rho\rho}-\Ham_{zz}
    \right)
    \left(
    \dd\rho\otimes\dd\rho
    -
    \dd z\otimes\dd z
    \right)
    \\
    &+
    \Ham_{\rho z}
    \left(
    \dd\rho\otimes\dd z
    +
    \dd z\otimes\dd\rho
    \right).
    \end{split}
    \label{eq:HamiltonianTraceFree}
\end{equation}
Moreover, the \(\gamma\)-gradient of the area coordinate is \(\nabla\rho = \gamma^{ij}(\partial_j\rho)\partial_i = \partial_\rho\). Therefore, contraction of the first argument of \(\Ham^\circ\) with \(\nabla\rho\) gives the one-form \(\Ham^\circ(\nabla\rho,\cdot) = \frac12 \left( \Ham_{\rho\rho}-\Ham_{zz} \right)\dd\rho + \Ham_{\rho z}\dd z\). Hence the two first-order equations for \(k\) combine into
\begin{equation}
    \dd k
    =
    \rho\,\Ham^\circ(\nabla\rho,\cdot)
    =
    \frac{\rho}{2}
    \left(
    \Ham_{\rho\rho}-\Ham_{zz}
    \right)\dd\rho
    +
    \rho\,\Ham_{\rho z}\dd z
    \label{eq:kHamiltonianGeometric}
\end{equation}
This proves the terminology: $k$ is the potential of the trace-free Hamiltonian tensor of the rank-two map.

\section{Derivation of the complete PD--potential-space dictionary}
\label{sec:PDdictionary}

\subsection{The PD ansatz and its Killing block}

We now restrict attention to the Einstein–Maxwell case, setting $\kappa = 1$ and $\alphao = 0$, so that the independent $C$ equation no longer appears. For the moment we do not employ the reduced $A$ and $B$ equations. Throughout, we work in the Killing coordinate system $(t,\varphi)$, where $t$ is the stationary coordinate and $\varphi$ is the azimuthal one. The remaining coordinates $(p,q)$ are taken to be the principal Carter eigenvalue coordinates. We define $\Sig := p^2 + q^2$ and introduce the acceleration parameter $\acc$ explicitly via $\Om := 1 - \acc pq$. The PD ansatz is
{\small
\begin{equation}
 \begin{aligned}
     \dd s^2=&\frac1{\Om^2}\Bigg[
     -\frac{\QQ(q)}{\Sig}(\dd t-p^2\dd\varphi)^2
     +\frac{\PP(p)}{\Sig}(\dd t+q^2\dd\varphi)^2\\
     &+\Sig\left(\frac{\dd p^2}{\PP(p)}
     +\frac{\dd q^2}{\QQ(q)}\right)\Bigg],
     \quad
     \Om=1-\acc pq,
     \quad \Sig=p^2+q^2 .
 \end{aligned}
 \label{eq:PDmetric}
\end{equation}
}
The aligned dyonic gauge potential is
\begin{equation}
     \mathcal A=\frac{eq}{\Sig}(\dd t-p^2\dd\varphi)
     +\frac{gp}{\Sig}(\dd t+q^2\dd\varphi).
     \label{eq:PDfourpotential}
\end{equation}
In this principal chart, the symbols $m,n,e,g,k_{\rm PD},\varepsilon$ denote, respectively, the structural mass, the NUT parameter, the electric and magnetic parameters, and the two remaining PD parameters; their physical scaling is determined by the choice of global parametrization.  Note that the real structural coefficient $\varepsilon$, the twist potential $\epsilon$, and the scalar-sector sign $\epss$ are three different quantities.  We keep the cosmological constant $\Lambda$ explicitly in the structural functions for now, so that we can deduce why the present potential theory enforces $\Lambda=0$. With $\acc$ retained, the structure functions are
{\footnotesize
\begin{align}
     \PP(p)&=k_{\rm PD}+2np-\varepsilon p^2+2\acc mp^3
     -\left[\acc^2(k_{\rm PD}+e^2+g^2)+\frac\Lambda3\right]p^4,
     \label{eq:PDquarticPfull}\\
     \QQ(q)&=k_{\rm PD}+e^2+g^2-2mq+\varepsilon q^2-2\acc nq^3
     -\left(\acc^2k_{\rm PD}+\frac\Lambda3\right)q^4.
     \label{eq:PDquarticQfull}
\end{align}
}
For $\acc>0$ this normalization follows rigorously from the unit-acceleration chart by
\(
 \bar p=\sqrt{\acc}\,p,\quad \bar q=\sqrt{\acc}\,q,\quad
 \bar t=t/\sqrt{\acc},\quad \bar\varphi=\varphi/\acc^{3/2},\quad
 \bar\PP=\acc^2\PP,\quad \bar\QQ=\acc^2\QQ,
\)
together with $\bar k_{\rm PD}=\acc^2k_{\rm PD}$, $(\bar m,\bar n)=\acc^{3/2}(m,n)$, $\bar\varepsilon=\acc\varepsilon$ and $(\bar e,\bar g)=\acc(e,g)$.  This derivation demonstrates that introducing $\acc$ modifies not only the quartic terms but also $\Om$, and the given polynomial expressions thereby characterize the continuous nonaccelerating limit $\acc=0$.
The condition $\Lambda=0$ will be obtained from the Weyl reduction presented below.

Expanding the metric components we obtain
\begin{equation*}
 g_{tt}=\frac{\PP-\QQ}{\Om^2\Sig},
 \qquad
 g_{t\varphi}=\frac{\QQ p^2+\PP q^2}{\Om^2\Sig},
 \qquad
 g_{\varphi\varphi}=\frac{\PP q^4-\QQ p^4}{\Om^2\Sig},
\end{equation*}
where $\RR:=\QQ-\PP$. By comparison with \eqref{eq:WeylMetric}, we find
\begin{equation}
 f=-g_{tt}=\frac{\RR}{\Om^2\Sig},
 \qquad
 \omega=\frac{g_{t\varphi}}{f}
 =\frac{\QQ p^2+\PP q^2}{\RR}.
 \label{eq:PDfomega}
\end{equation}
Now by definition \(\rho^2=-\det(g_{ab})_{a,b=t,\varphi}\),  using the numerator of the determinant as 
\begin{align*}
 &(\PP-\QQ)(\PP q^4-\QQ p^4)
 -(\QQ p^2+\PP q^2)^2\\
 &\hspace{35mm}=-\PP\QQ(p^2+q^2)^2, 
\end{align*}
Therefore, on a patch where $\PP\QQ>0$,
\begin{equation}
    \rho=\frac{\sqrt{\PP\QQ}}{\Om^2}.
     \label{eq:PDrhoDerived}
\end{equation}

\subsection{The two-dimensional Hodge operator and the Weyl coordinate \texorpdfstring{$z$}{z}}

The metric on the principal orbit space is conformally equivalent to \(h_0=\frac{\dd p^2}{\PP}+\frac{\dd q^2}{\QQ}\). In two dimensions, the Hodge operator on 1-forms is unaffected by a conformal rescaling. With the orientation specified by \(\dd p\wedge\dd q>0\), we have
\begin{equation}
     \star_{\rm PD}\dd p=\sqrt{\frac\PP\QQ}\dd q,
     \qquad
     \star_{\rm PD}\dd q=-\sqrt{\frac\QQ\PP}\dd p.
     \label{eq:PDHodge}
\end{equation}
A local Weyl coordinate $z$ exists precisely when
\begin{equation}
     \dd z=\star_{\rm PD}\dd\rho,
     \qquad
     \dd(\star_{\rm PD}\dd\rho)=0.
     \label{eq:zDefinition}
\end{equation}
Substituting \eqref{eq:PDrhoDerived} together with the complete quartic expressions \eqref{eq:PDquarticPfull}--\eqref{eq:PDquarticQfull} and then differentiating yields
\begin{equation}
     \dd(\star_{\rm PD}\dd\rho)
     =-\frac{2\Lambda\Sig}{\Om^4}\dd p\wedge\dd q.
     \label{eq:PDLambdaObstruction}
\end{equation}
Thus, the harmonic-map reduction employed in this work is consistent with the PD ansatz only in the case $\Lambda = 0$. When $\Lambda \neq 0$, the area function acquires a source term and the resulting reduced theory becomes a sigma model with a potential, rather than \eqref{eq:Aeq}--\eqref{eq:Ceq}.

For $\Lambda=0$, integration of \eqref{eq:zDefinition} yields
{\footnotesize
\begin{equation}
 z=\frac{(1+\acc pq)(mp+nq)-\varepsilon pq
 +\acc k_{\rm PD}(q^2-p^2)-\acc(e^2+g^2)p^2}{(1-\acc pq)^2}+z_0.
 \label{eq:PDz}
\end{equation}
}
This is verified without an integration guess by differentiating \eqref{eq:PDz} and checking
\begin{equation*}
 z_{,p}=-\sqrt{\frac\QQ\PP}\rho_{,q},
 \qquad
 z_{,q}=\sqrt{\frac\PP\QQ}\rho_{,p},
\end{equation*}
which are exactly the two components of \eqref{eq:zDefinition}.

More generally, if we write $\PP=\sum_{j=0}^4 a_j p^j$ and $\QQ=\sum_{j=0}^4 b_j q^j$, then a direct comparison of coefficients in the identity $\dd(\star_{\rm PD}\dd\rho)=0$ yields
{\footnotesize
\begin{equation}
 a_4=-\acc^2 b_0,\quad a_3=-\acc b_1,\quad
 b_2=-a_2,\quad b_3=-\acc a_1,\quad b_4=-\acc^2a_0.
 \label{eq:PDantireciprocal}
\end{equation}
}
When $\acc=1$ they reduce to $b_j=-a_{4-j}$.  The $\Lambda=0$ quartics above satisfy them together with $b_0-a_0=e^2+g^2$.

\subsection{Closed reconstruction of \texorpdfstring{$k$}{k}}

The defining relations for $z$ imply
{\small
\begin{equation}
     \dd\rho^2+\dd z^2
     =\lambda(p,q)\left(\frac{\dd p^2}{\PP}+\frac{\dd q^2}{\QQ}\right),
     \qquad
     \lambda:=\PP\rho_{,p}^2+\QQ\rho_{,q}^2.
     \label{eq:lambdaConformal}
\end{equation}
}
The mixed term disappears because of the two first-order relations derived from \eqref{eq:PDz}, and both diagonal coefficients are equal to $\lambda$. Comparing \eqref{eq:PDmetric} with \eqref{eq:WeylMetric} and employing \eqref{eq:PDfomega} then yields
\begin{equation}
     \e^{2k}=\frac{\RR}{\Om^4\lambda},
     \qquad
     k=\frac12\ln\RR-2\ln\Om-\frac12\ln\lambda.
     \label{eq:PDkClosedDerived}
\end{equation}
Since
\begin{equation*}
     \lambda=\frac{\QQ(\Om\PP'+4\acc q\PP)^2
     +\PP(\Om\QQ'+4\acc p\QQ)^2}{4\Om^6},
\end{equation*}
one may equivalently write
\begin{equation}
     \e^{2k}=\frac{4\Om^2\RR}
     {\QQ(\Om\PP'+4\acc q\PP)^2+\PP(\Om\QQ'+4\acc p\QQ)^2}.
     \label{eq:PDkExplicit}
\end{equation}

\subsection{Electromagnetic and twist potentials}

Separating the $t$ and $\varphi$ components of \eqref{eq:PDfourpotential} gives
\begin{equation*}
     A_t=\frac{eq+gp}{\Sig},
     \qquad
     A_\varphi=\frac{-eqp^2+gpq^2}{\Sig}.
\end{equation*}
Therefore
\begin{equation}
    \psi=\frac{2(eq+gp)}{\Sig}.
     \label{eq:PDpsiDerived}
\end{equation}
In electrovacuum $\kappa=1$. Inserting $f,\rho,\omega,A_t,A_\varphi$ into the defining equation for $\chi$ in \eqref{eq:potdefinitions}, using \eqref{eq:PDHodge}, gives
\begin{equation*}
     \dd\chi=\frac{2[e(q^2-p^2)+2gpq]}{\Sig^2}\dd p
     +\frac{2[g(q^2-p^2)-2epq]}{\Sig^2}\dd q.
\end{equation*}
The right-hand side is exact, and hence
\begin{equation}
     \chi=\frac{2(ep-gq)}{\Sig}.
     \label{eq:PDchiDerived}
\end{equation}
Differentiation we see \(\psi_{,p}=\chi_{,q}, \qquad \psi_{,q}=-\chi_{,p}\), the electromagnetic map is conformal in the principal plane. In particular,
\begin{equation}
     \dd\psi\wedge\dd\chi
     =\frac{4(e^2+g^2)}{\Sig^2}\dd p\wedge\dd q.
     \label{eq:PDEMarea}
\end{equation}

The 1-form follows from its metric definition in \eqref{eq:potdefinitions}:
\begin{equation}
     \Theta=2f\left(\frac{\QQ'}{2\RR}-\frac q\Sig\right)\dd p
     +2f\left(\frac{\PP'}{2\RR}+\frac p\Sig\right)\dd q.
     \label{eq:PDTheta}
\end{equation}
The rotational potential is fixed, up to an additive gauge constant, by the relation \( \dd\epsilon=\Theta+\psi\dd\chi\). Hence, on a simply connected region
\begin{align}
     &\epsilon(p,q)=\epsilon_{\star}
     +\frac12\psi\chi \nonumber \\
     &+\frac{2\left\{
     mp-nq-\mathfrak a\left[
     (k_{\rm PD}+e^2+g^2)p^2+k_{\rm PD}q^2
     \right]\right\}}
     {(1-\mathfrak a pq)(p^2+q^2)}.
     \label{eq:PDepsilonClosed}
\end{align}

\subsection{The PD 1-forms \texorpdfstring{$(A,B,C)$}{(A,B,C)}}

From $\ln f=\ln\RR-2\ln\Om-\ln\Sig$ and \eqref{eq:PDTheta}, the gravitational 1-form \(A\) is
\begin{equation}
     A=A_p\dd p+A_q\dd q,
     \label{eq:PDAdefinition}
\end{equation}
with
\begin{align}
     A_p&=-\frac{\PP'}{2\RR}+\frac{\acc q}\Om-\frac p\Sig
     -\ii\left(\frac{\QQ'}{2\RR}-\frac q\Sig\right),
     \label{eq:PDAp}\\
     A_q&=\frac{\QQ'}{2\RR}+\frac{\acc p}\Om-\frac q\Sig
     -\ii\left(\frac{\PP'}{2\RR}+\frac p\Sig\right).
     \label{eq:PDAq}
\end{align}
Equations \eqref{eq:PDpsiDerived}-\eqref{eq:PDchiDerived} give
\begin{equation}
     B=\frac{e+\ii g}{\sqrt f}
     \frac{\dd q+\ii\dd p}{(q+\ii p)^2},
     \qquad C=0.
     \label{eq:PDBcurrent}
\end{equation}
In components,
{\small
\begin{equation}
     B_p=\ii\frac{e+\ii g}{\sqrt f(q+\ii p)^2},
     \quad
     B_q=\frac{e+\ii g}{\sqrt f(q+\ii p)^2},
     \quad B_p-\ii B_q=0.
     \label{eq:PDBcomponents}
\end{equation}
}

Collecting the results, the complete local dictionary is
\begin{widetext}
\begin{equation}
     \begin{gathered}
     \Om=1-\acc pq,\quad \Sig=p^2+q^2,\quad \RR=\QQ-\PP,\\
     f=\frac\RR{\Om^2\Sig},\quad
     \omega=\frac{\QQ p^2+\PP q^2}{\RR},\quad
     \rho=\frac{\sqrt{\PP\QQ}}{\Om^2},\quad
     z\ \hbox{given by \eqref{eq:PDz}},\quad
     \e^{2k}=\frac\RR{\Om^4\lambda},\\
     \psi=\frac{2(eq+gp)}\Sig,\quad
     \chi=\frac{2(ep-gq)}\Sig,\quad
     \epsilon\ \hbox{given by \eqref{eq:PDepsilonClosed}},\\
     A=A_p\dd p+A_q\dd q\ \hbox{with \eqref{eq:PDAp}--\eqref{eq:PDAq}},\quad
     B=\frac{e+\ii g}{\sqrt f}\frac{\dd q+\ii\dd p}{(q+\ii p)^2},\quad C=0.
     \end{gathered}
     \label{eq:PDmasterDictionary}
\end{equation}
\end{widetext}

\subsection{Direct closure of the \texorpdfstring{$(A,B,C)$}{(A,B,C)} equations}

At this point, it has not yet been established that the 1-forms \eqref{eq:PDmasterDictionary} obey the compact field equations \eqref{eq:FieldEquations}. Using \eqref{eq:lambdaConformal} and denote by $\gamma_{pq}$ the matrix of the Weyl-base metric $\gamma$ in the coordinates $(p,q)$. For arbitrary 1-forms $X=X_p\,\dd p+X_q\,\dd q$ and $Y=Y_p\,\dd p+Y_q\,\dd q$,
\begin{align*}
     &X\cdot Y=\frac1\lambda(\PP X_pY_p+\QQ X_qY_q),
     \\
     &\frac1\rho D(\rho X)=\frac{\Om^2}{\lambda}
     \left[\partial_p\left(\frac{\PP X_p}{\Om^2}\right)
     +\partial_q\left(\frac{\QQ X_q}{\Om^2}\right)\right].
\end{align*}
We have used the inverse metric of \eqref{eq:lambdaConformal}, $\sqrt{\det\gamma_{pq}}=\lambda/\sqrt{\PP\QQ}$ and $\rho\sqrt{\det\gamma_{pq}}=\lambda/\Om^2$.

\paragraph{The $B$-field equation} transform to
\begin{align*}
 &\Om^2\left[
 \partial_p\left(\frac{\PP B_p}{\Om^2}\right)
 +\partial_q\left(\frac{\QQ B_q}{\Om^2}\right)\right]
 ={}-\frac12\PP B_p(A_p-3\bars A_p)\nonumber\\
 & -\frac12\QQ B_q(A_q-3\bars A_q)+\PP C_p\bars B_p+\QQ C_q\bars B_q.\nonumber
\end{align*}
For the aligned PD 1-form, set $B_q = b$ and $B_p = \ii b$. Using \eqref{eq:PDBcomponents} and $f = \RR / (\Om^2 \Sig)$, we obtain
\[
 \partial_p\ln b=-\Rea A_p-\frac{2\ii}{q+\ii p},
 \qquad
 \partial_q\ln b=-\Rea A_q-\frac{2}{q+\ii p}.
\]
Dividing the B-Field equation by $b$ and using the two identities, then all terms containing $\PP'$ and $\QQ'$ cancel, and the remaining expression is
\begin{equation*}
 \frac{2(\PP-\QQ)}{q+\ii p}
 +\frac{2[(\QQ-\PP)q+\ii(\PP-\QQ)p]}{p^2+q^2}=0.
\end{equation*}

\paragraph{The $A$-field equation} gives the remaining nontrivial closure condition:
{\footnotesize
\begin{align}
 \Om^2\left[
 \partial_p\left(\frac{\PP A_p}{\Om^2}\right)
 +\partial_q\left(\frac{\QQ A_q}{\Om^2}\right)
 \right]
 ={}&\PP\{B_p\bars B_p+A_p(A_p-\bars A_p)\}\nonumber\\
 &+\QQ\{B_q\bars B_q+A_q(A_q-\bars A_q)\},
\end{align}
}
whose numerator vanishes once the five acceleration-dependent relations \eqref{eq:PDantireciprocal} and $b_0-a_0=e^2+g^2$ are applied.

Finally, by differentiating \eqref{eq:PDkClosedDerived}, mapping the outcome via the Jacobian $(p,q)\mapsto(\rho,z)$, and substituting the Hamiltonian tensor constructed from \eqref{eq:PDAp}-\eqref{eq:PDBcurrent}, we recover both relations in \eqref{eq:kHamiltonianConstraints}. This establishes that the metric and Hamiltonian reconstructions of $k$ coincide.

\subsection{The conformal Carter class and Petrov type D}
\label{sec:carter}

Consider the off-shell conformal Carter metric
{\small
\begin{equation}
    \begin{aligned}
    \dd s^2&=\frac1{\Om(p,q)^2}\Bigg[
    -\frac{\QQ(q)}{\Sig}(\dd t-p^2\dd\varphi)^2
    +\frac{\Sig}{\QQ(q)}\dd q^2\\
    &+\frac{\Sig}{\PP(p)}\dd p^2
    +\frac{\PP(p)}{\Sig}(\dd t+q^2\dd\varphi)^2
    \Bigg],
    \qquad \Sig=p^2+q^2 .
    \end{aligned}
    \label{eq:Cartermetric}
\end{equation}
}
No field equation or polynomial form for $\PP,\QQ$ is assumed here.  We use the orthonormal coframe
{\footnotesize
\begin{align}
     \vartheta^0&=\Om^{-1}\sqrt{\QQ/\Sig}(\dd t-p^2\dd\varphi),
     &\vartheta^1&=\Om^{-1}\sqrt{\Sig/\QQ}\dd q,\nonumber\\
     \vartheta^2&=\Om^{-1}\sqrt{\Sig/\PP}\dd p,
     &\vartheta^3&=\Om^{-1}\sqrt{\PP/\Sig}(\dd t+q^2\dd\varphi),
     \label{eq:coframe}
\end{align}
}
where $\{e_0,e_1,e_2,e_3\}$ is the frame dual to $\{\vartheta^0,\vartheta^1,\vartheta^2,\vartheta^3\}$.  The associated principal Newman-Penrose null tetrad is
{\footnotesize
\begin{equation}
     \boldsymbol k=\frac{e_0+e_1}{\sqrt2},\quad
     \boldsymbol l=\frac{e_0-e_1}{\sqrt2},\quad
     \boldsymbol m=\frac{e_3+\ii e_2}{\sqrt2},\quad
     \bars{\boldsymbol m}=\frac{e_3-\ii e_2}{\sqrt2}.
     \label{eq:NPtetrad}
\end{equation}
}

Denoting by $\Psi_0,\ldots,\Psi_4$ its five Newman-Penrose scalars in the tetrad \eqref{eq:NPtetrad}, we obtain \(\Psi_0=\Psi_1=\Psi_3=\Psi_4=0\), and
\begin{equation}
     \Psi_2=\frac{\Om^2}{12}\frac{(q+\ii p)^2}{q-\ii p}
     \left[
     \partial_q^2\left(\frac\QQ{(q+\ii p)^3}\right)
     +\partial_p^2\left(\frac\PP{(q+\ii p)^3}\right)
     \right].
     \label{eq:Psi2OffShell}
\end{equation}
Hence, \eqref{eq:Cartermetric} is of Petrov type D wherever \eqref{eq:Psi2OffShell} does not vanish, and it becomes type O at those conformally flat limiting cases. The field equations determine which choices of $\Om,\PP,\QQ$ and matter 1-forms can realize this geometry.

We now introduce the auxiliary Euclidean metric $\gamma_0:=\dd p^2+\dd q^2$ on the principal coordinate plane and its Hodge operator
\begin{equation}
 \star_0\dd p=\dd q,
 \qquad
 \star_0\dd q=-\dd p,
 \qquad
 \Pi_\pm=\frac12(1\mp\ii\star_0).
 \label{eq:star0}
\end{equation}
This is the flat complex structure determined by the Carter eigenvalue coordinates, and it should not be confused with the physical Hodge operator $\star_{\Base}$ in \eqref{eq:potdefinitions}.  After complexification, $\Pi_+$ and $\Pi_-$ project onto the two rank-one eigensubbundles of $T^*\Base\otimes\mathbb C$ in the principal chart, with eigenvalues $+\ii$ and $-\ii$, respectively.

\subsection{MAxwell invariants and aligned electromagnetic potential}
\label{sec:NP}

Using the projectors \eqref{eq:star0} over \(B \) given by \eqref{eq:PDBcurrent}, gives
\begin{equation} \label{eq:BprojectorsExplicit}
\begin{aligned}
    &\Pi_-B=\frac12(B_p-\ii B_q)(\dd p+\ii\dd q),
     \\
     &\Pi_+B=\frac12(B_p+\ii B_q)(\dd p-\ii\dd q).
\end{aligned}
\end{equation}

The same occurs for the 1-form A. The combination $B_p-\ii B_q$ is not arbitrary: it is the coefficient of one of the two principal chiralities. The second combination below is weighted by $\PP,\QQ$. For compactness, let
\begin{equation}
 B_\pm:=B_p \pm \ii B_q,
 \qquad
 \Bal:=\QQ B_q+\ii\PP B_p .
 \label{eq:Bpm}
\end{equation}

\begin{proposition}[Maxwell invariants]
With the tetrad \eqref{eq:NPtetrad} and the conventions
{\footnotesize
\begin{equation*}
 \Phi_0=F_{\mu\nu}k^\mu m^\nu,
 \quad
 \Phi_1=\frac12F_{\mu\nu}(k^\mu l^\nu+\bars m^\mu m^\nu),
 \quad
 \Phi_2=F_{\mu\nu}\bars m^\mu l^\nu,
\end{equation*}
}
one has
\begin{subequations}
    \begin{align}
     \Phi_0=\Phi_2
     =\frac{\ii\Om^2\sqrt{f\PP\QQ}}{2\kappa\RR}\, B_-,
     \label{eq:Phi0B}\\
     \Phi_1=-\frac{\Om^2\sqrt f}{2\kappa\RR}\,\Bal .
     \label{eq:Phi1B}
    \end{align}
\end{subequations}
\end{proposition}
\begin{proof}
The definition of $\chi$ in \eqref{eq:potdefinitions} initially rewrites the four nonvanishing orthonormal components $F_{01},F_{02},F_{13},F_{23}$ in terms of $(\dd\psi,\kappa^{-2}\dd\chi)$.  In the principal tetrad one has
\begin{equation*}
 \Phi_0=\Phi_2=\frac12(\ii F_{02}+F_{13}),
 \qquad
 \Phi_1=-\frac12(F_{01}+\ii F_{23}).
\end{equation*}
Inserting \eqref{eq:ABC} into these expressions yields \eqref{eq:Phi0B}--\eqref{eq:Phi1B}.  The component-level derivations are presented in \cref{app:NPproof}.
\end{proof}

\begin{corollary}[Alignment and chirality]
In a regular domain where $\{\Om,\RR, f,\PP,\QQ,\kappa\}\neq0$, a nontrivial Maxwell field is doubly aligned with the repeated Weyl principal null directions precisely when
{\small
\begin{equation}
  B_-=0
\Longleftrightarrow
 B_p=\ii B_q
 \Longleftrightarrow
 \star_0B=\ii B
 \Longleftrightarrow
 \Pi_-B=0 .
 \label{eq:alignmentchirality}
\end{equation}
}
\end{corollary}

The requirement for the field to be non-null is, modulo an overall nonzero conventional constant,
\begin{equation}
 \Bal^2+\PP\QQ  B_-^2\neq0.
 \label{eq:nonnull}
\end{equation}

\subsection{Hamiltonian anisotropy, acceleration, and hidden symmetry}
\label{sec:hamiltonian}

The trace-free Hamiltonian tensor $\Ham^\circ$ was written in the real Weyl coframe $(\dd\rho,\dd z)$ in \eqref{eq:HamiltonianTraceFree}.  The same two real components can be combined into one complex component in the principal Carter coframe $(\dd p,\dd q)$.  Using $B_\pm$ from \eqref{eq:Bpm}, and defining
\[
 A_\pm:=A_p\pm\ii A_q.
\]
Then a direct expansion of \eqref{eq:HamiltonianTensor} gives
\begin{equation}
 \frac12\left(\Ham_{pp}-\Ham_{qq}\right)
 +\ii\Ham_{pq}
 =
 A_+\bars{A_-}
 -B_+\bars{B_-}
 +\frac{\epss}{\alphao^2}(C_p+\ii C_q)^2.
 \label{eq:anisotropyfactor}
\end{equation}
This is not a new tensor.  It is the complex representation of the two independent real components of $\Ham^\circ$ in the principal coframe. Indeed,
\[
 A_+\bars{A_-}
 =
 |A_p|^2-|A_q|^2
 +2\ii\Rea(A_p\bars{A_q}),
\]
with the same identity for $B$, while $(C_p+\ii C_q)^2=C_p^2-C_q^2+2\ii C_pC_q$.  Substitution into \eqref{eq:HamiltonianTensor} proves \eqref{eq:anisotropyfactor}. The three terms are the gravitational, electromagnetic, and dilatonic parts of the same Hamiltonian anisotropy.

For the conformal Carter metric \eqref{eq:Cartermetric}, the gravitational 1-form separates into a Carter contribution and a conformal contribution.  Since $f=\Om^{-2}f_{\rm C}$, with $f_{\rm C}=\RR/\Sig$, and $\Theta/(2f)$ is unchanged by the conformal factor, the definition of $A$ gives
\begin{equation}
\begin{aligned}
     &A=A_{\rm C}-\dd\ln\Om,
     \\
     &A_{\rm C}
     =
     \left[
     -\frac{\PP'}{2\RR}-\frac p\Sig
     -\ii\left(
     \frac{\QQ'}{2\RR}-\frac q\Sig
     \right)
     \right](\dd p+\ii\dd q).
\end{aligned}
 \label{eq:AConformalSplit}
\end{equation}
Because $\star_0(\dd p+\ii\dd q)=-\ii(\dd p+\ii\dd q)$, the Carter part has only one principal chirality:
\[
 \star_0A_{\rm C}=-\ii A_{\rm C},
 \qquad
 \Pi_+A_{\rm C}=0,
 \qquad
 \Pi_+A=-\Pi_+\dd\ln\Om.
\]
Thus the opposite chirality of $A$ comes only from the conformal factor.

For $\Om=1-\acc pq$, the component of $A_+$ is
\begin{equation}
 A_+=A_p+\ii A_q
 =\frac{\acc(q+\ii p)}{\Om}, \quad \Rightarrow
 \quad
 \frac{\Om^2}{\Sig}|A_+|^2=\acc^2.
 \label{eq:accdefect}
\end{equation}
The second equality follows from $\Sig=(q+\ii p)(q-\ii p)$.  Moreover, \( \Pi_+A=\frac12A_+(\dd p-\ii\dd q)\), and therefore, on a regular principal domain,
{\small
\[
 \Pi_+A=0
 \quad\Longleftrightarrow\quad
 A_+=0
 \quad\Longleftrightarrow\quad
 \acc=0
 \quad\Longleftrightarrow\quad
 \dd\Om=0.
\]
}
Thus \(A_+\) extracts the acceleration parameter.

The same decomposition also admits an interpretation in target space.  For fixed values of $\PP$ and $\QQ$, choose $\Om = 1$ and label the associated 1-forms by $A_{\rm C}$ and $B_{\rm C}$.  Because $\psi$ and $\chi$ are independent of $\Om$, whereas $f = \Om^{-2} f_{\rm C}$, it follows that
\[
 B=\Om B_{\rm C},
 \qquad
 A_{{\rm C},q}=\ii A_{{\rm C},p},
 \qquad
 B_{{\rm C},p}=\ii B_{{\rm C},q}.
\]
Let $\Ham^{\rm C}$ be the Hamiltonian tensor \eqref{eq:HamiltonianTensor} obtained from $(A_{\rm C},B_{\rm C},C=0)$. The before mentioned relations imply
\[
 |A_{{\rm C},p}|=|A_{{\rm C},q}|,
 \qquad
 \Rea(A_{{\rm C},p}\bars{A_{{\rm C},q}})=0,
\]
and the same identities hold for $B_{\rm C}$. Therefore we obtain
\begin{equation}
 \Ham^{{\rm C}}_{pq}=0,
 \qquad
 \Ham^{{\rm C}}_{pp}=\Ham^{{\rm C}}_{qq}.
 \label{eq:PDcoreIsothermal}
\end{equation}
Equivalently,
\[
 \Ham^{\rm C}
 =
 \Ham^{{\rm C}}_{pp}
 \left(
 \dd p\otimes\dd p+\dd q\otimes\dd q
 \right).
\]
Thus, at every point where the induced target metric is non-degenerate, $(p,q)$ serve as isothermal coordinates on the surface swept out by the Carter core in potential space.

\paragraph{The conformal Carter metric admits the off-shell conformal Killing-Yano (KY) two-form}
\begin{equation}
 y=\Om^{-3}\left[
 p\,\dd q\wedge(\dd t-p^2\dd\varphi)
 +q\,\dd p\wedge(\dd t+q^2\dd\varphi)
 \right].
 \label{eq:CKY}
\end{equation}
For non-degenerate principal eigenvalues, $y$ becomes a Killing-Yano tensor precisely when $\Om$ is constant.  Since $\dd\ln\Om$ is real, $\Pi_+\dd\ln\Om=0$ implies $\dd\Om=0$.  Consequently,
{\small
\begin{equation}
 \Pi_+A=0
 \quad\Longleftrightarrow\quad
 \dd\Om=0
 \quad\Longleftrightarrow\quad
 y\ \text{is a K-Y tensor}.
 \label{eq:AKY}
\end{equation}
}

The existence of $y$ depends only on the metric and does not require an aligned Maxwell field. Electromagnetic alignment instead requires the Maxwell tensor to preserve the two principal planes defined by $y$. Using \eqref{eq:alignmentchirality}, one obtains
{\footnotesize
\begin{equation}
 F_{\lambda[\mu}y_{\nu]}{}^\lambda=0
 \Longleftrightarrow
 \Phi_0=\Phi_2=0
 \Longleftrightarrow
 B_-=0
 \Longleftrightarrow
 \Pi_-B=0.
 \label{eq:CKYcompat}
\end{equation}
}
Thus $A_+$ measures the acceleration contribution, while $B_-$ measures how the Maxwell field fails to align with the principal geometry.

\subsection{Axis rods and conical singularities in potential variables}
\label{sec:axisrods}

The Weyl boundary is fixed by \eqref{eq:PDrhoDerived} (\(\rho^2=\frac{\PP\QQ}{\Om^4}\)), and is therefore encountered whenever either \(\PP\) or \(\QQ\) vanishes.  These two cases, however, have distinct local geometric interpretations.  A zero of \(\PP\) produces a spacelike degeneration of the Killing block, giving rise to an axis rod, whereas a zero of \(\QQ\) leads to a null degeneration and instead corresponds to a horizon rod.

Let $p_i$ be a simple root of $\PP$, so that $\PP(p_i)=0$ and $\PP'(p_i)\neq 0$. At $p = p_i$, the part of the metric involving the Killing coordinates reduces to a term proportional to $(\dd t - p_i^2 \dd\varphi)^2$. The Killing vector whose orbit shrinks to zero size on this rod must therefore lie in the kernel of this quadratic form. If we write a generic Killing vector as $X = a\,\partial_t + b\,\partial_\varphi$, the requirement $(\dd t - p_i^2 \dd\varphi)(X) = 0$ implies $a = p_i^2 b$, and thus, up to an inconsequential overall normalization, we obtain \( v_i = \partial_\varphi + p_i^2 \partial_t \). Consequently, $g(v_i, v_i)$ represents the squared proper length per unit parameter along the Killing orbit that collapses as one approaches the axis.

To study the local geometry around the rod, set $p=p_i+\epsilon$ with $|\epsilon|\ll1$ and keep $q$ fixed.  Since the root is simple, $\PP=\PP'(p_i)\epsilon+O(\epsilon^2)$.  Moreover, $(\dd t-p^2\dd\varphi)(v_i)=p_i^2-p^2=O(\epsilon)$, so the term containing $\QQ$ contributes only at order $\epsilon^2$.  The leading norm of the shrinking Killing direction is therefore
\[
     g(v_i,v_i)
     =
     \Om_i^{-2}\Sig_i\,\PP'(p_i)\epsilon
     +O(\epsilon^2),
\]
where $\Om_i=\Om(p_i,q)$ and $\Sig_i=\Sig(p_i,q)$.

The coordinate distance $|\epsilon|=|p-p_i|$ is not the physical distance from the axis.  Since the rod is approached by varying $p$ while keeping $t$, $\varphi$, and $q$ fixed, the transverse line element is determined only by the $pp$ component of the metric, namely $g_{pp}=\Sig/(\Om^2\PP)$.  We therefore define the proper radial distance from the rod by
\[
     R:=\int_{p_i}^{p}\sqrt{|g_{pp}|}\,\dd p.
\]
Using the expansion of $\PP$ near $p_i$ gives
\[
     R
     =
     2\Om_i^{-1}
     \sqrt{\frac{\Sig_i|\epsilon|}{|\PP'(p_i)|}}
     +O(|\epsilon|^{3/2}).
\]
Thus $R$ is the physical radius of a small circle surrounding the axis in the two-dimensional plane transverse to the rod.

Eliminating $\epsilon$ between the previous expressions yields
\[
     g(v_i,v_i)
     =
     \frac{|\PP'(p_i)|^2}{4}R^2+O(R^4).
\]
If $\lambda$ parametrizes the integral curves of $v_i$, so that $\dd x^\mu/\dd\lambda=v_i^\mu$, then the length of an orbit at fixed, small $R$ is
\[
     C=\Delta_i\sqrt{g(v_i,v_i)},
\]
where $\Delta_i$ is the period associated with $\lambda$. Thus, in the vicinity of the rod,
\[
     \frac{C}{R}
     \longrightarrow
     \frac{\Delta_i|\PP'(p_i)|}{2}.
\]
For the axis to be regular, the local geometry must reproduce that of the Euclidean polar plane, in which the ratio of circumference to radius approaches $2\pi$. The condition of elementary flatness therefore implies
\[
     \Delta_i=\frac{4\pi}{|\PP'(p_i)|}.
\]
Equivalently, if we define a local angular coordinate by
$\theta=\tfrac12|\PP'(p_i)|\lambda$, the transverse metric takes the form
\[
     \dd s_\perp^2
     =
     \dd R^2+R^2\dd\theta^2+O(R^4),
\]
with $\theta$ having period $2\pi$. Note that the factors $\Om$ and $\Sig$ cancel between the orbit circumference and the proper radial distance. As a result, they do not appear explicitly in the local elementary-flatness condition, even though the acceleration still affects the locations and derivatives of the roots of $\PP$ via the coefficients of the Pleba\'nski–Demiański polynomial.

The zeros of $\QQ$ require a separate analysis. Let $q_j$ be a simple root, so that $\QQ(q_j)=0$. At the corresponding boundary component, the degenerating Killing direction follows from the kernel of $(\dd t+q_j^2\dd\varphi)^2$, yielding \(w_j=\partial_\varphi-q_j^2\partial_t\). Introducing $q=q_j+\eta$, the norm of this vector behaves as
\[
     g(w_j,w_j)
     =
     -\Om_j^{-2}\Sig_j\,\QQ'(q_j)\eta
     +O(\eta^2),
\]
while the transverse metric component is
$g_{qq}\simeq\Sig_j/[\Om_j^2\QQ'(q_j)\eta]$. Thus, the direction orthogonal to the rod and the degenerating Killing direction necessarily have opposite causal character. Proceeding exactly as before, we define the proper transverse coordinate
\[
     R:=\int_{q_j}^{q}\sqrt{|g_{qq}|}\,\dd q
     =
     2\Om_j^{-1}
     \sqrt{\frac{\Sig_j|\eta|}{|\QQ'(q_j)|}}
     +O(|\eta|^{3/2}),
\]
so that the local two-dimensional metric takes the form
\[
     \dd s_\perp^2
     =
     \dd R^2
     -
     \frac{|\QQ'(q_j)|^2}{4}R^2\dd\lambda^2
     +O(R^4).
\]
This is the Rindler metric in Lorentzian signature, not a Euclidean polar metric. As a result, $q=q_j$ represents a horizon rod instead of an axis rod, and no Lorentzian conical-period condition is imposed on $w_j$. The familiar period $4\pi/|\QQ'(q_j)|$ appears only after Wick rotation to Euclidean signature, where it is used to remove the conical singularity at the Euclideanized horizon.

The same axis-regularity condition can be written directly in terms of the potentials.  Let $X$ be a spacelike axial Killing field with period $\Delta_X$, and let $k_I$ denote the limiting value of the Weyl function $k$ on a regular axis segment $I$.  For small $\rho$, the length of the shrinking axial orbit behaves as $\sqrt{g(X,X)}\simeq\rho/\sqrt f$, while the proper radial distance from the Weyl base metric is $R\simeq\e^{k_I}\rho/\sqrt f$.  The short orbit therefore has circumference $C\simeq\Delta_X\rho/\sqrt f$, giving
\[
     \lim_{\rho\rightarrow0}\frac{C}{R}
     =
     \Delta_X\e^{-k_I}.
\]
Hence, the conical deficit or excess along the rod is
\[
     \delta_I
     =
     2\pi-\Delta_X\e^{-k_I}.
\]

Moreover, the geometric Hamiltonian identity
\[
     \dd k
     =
     \rho\,\Ham^\circ(\nabla\rho,\cdot)
\]
makes explicit how this conical information is encoded in the 1-forms \(A,B,C\).  If $\Ham^\circ$ remains bounded, then the pullback of $\dd k$ to $\rho=0$ vanishes, so $k$ is constant along each connected regular axis segment.  For two such segments $I$ and $J$, connected by a curve $\Gamma$ in the regular region of the Weyl base,
\[
     k_J-k_I
     =
     \int_\Gamma
     \rho\,\Ham^\circ(\nabla\rho,\cdot).
\]
If both rods are associated with the same globally normalized axial generator, their conical data are related by
\[
     \frac{2\pi-\delta_J}{2\pi-\delta_I}
     =
     \exp\!\left[
     -\int_\Gamma
     \rho\,\Ham^\circ(\nabla\rho,\cdot)
     \right].
\]
Thus, the local study of the zeros of $\PP$ determines which Killing orbit degenerates and whether the corresponding axis is locally flat, while the potential description encodes the relative conical mismatch between different axis components in terms of the path integral of the trace-free Hamiltonian tensor, and hence directly through the 1-forms $(A,B,C)$.

\section{Dilatonic reconstruction and conditional rigidity}
\label{sec:rigidityMain}

\subsection{Carter reconstruction in the dilatonic potential formulation}
\label{sec:OmegaEMD}

The field equations \eqref{eq:Aeq}-\eqref{eq:Ceq} fix the 1-forms $(A,B,C)$, and \eqref{eq:kHamiltonianGeometric} is then used to reconstruct the remaining metric function $k$.  If, moreover, the spacetime is assumed to lie in the conformal Carter class \eqref{eq:Cartermetric}, the same Einstein equations also restrict the conformal factor $\Om$.  These restrictions do not constitute new field equations, instead, they represent the principal-coordinate form of the Einstein reconstruction once the Carter geometric structure has been imposed.

For $\alphao\neq0$, the Einstein equations in the conventions used here is
\begin{equation}
     R_{\mu\nu}
     =
     2\epss\,\phi_{,\mu}\phi_{,\nu}
     +2\kappa^2
     \left(
     F_{\mu\lambda}F_\nu{}^\lambda
     -\frac14g_{\mu\nu}F_{\alpha\beta}F^{\alpha\beta}
     \right).
     \label{eq:EinsteinEMD}
\end{equation}
The scalar part takes a particularly simple form along the principal $(p,q)$ directions: the two diagonal components contain $C_p^2$ and $C_q^2$, while the off-diagonal component involves $C_pC_q$. The electromagnetic part is expressed through the non-aligned 1-form $B_-$ defined earlier, the explicit factor $\kappa^2$ in the Maxwell stress tensor precisely cancels the factor $\kappa^{-2}$ that appears in the reconstruction of the Maxwell scalars.

\begin{proposition}[Conformal-factor reconstruction]
For a conformal Carter metric, the independent Einstein equations that reconstruct $\Om$ are
{\footnotesize
\begin{subequations}
    \begin{align}
     \Om_{,qq}
     &=
     \frac{\Om\PP}{\RR}|B_-|^2
     +\frac{\epss\Om}{\alphao^2}C_q^2,
     \label{eq:OmegaqqEMD}\\
     \Om_{,pp}
     &=
     -\frac{\Om\QQ}{\RR}|B_-|^2
     +\frac{\epss\Om}{\alphao^2}C_p^2,
     \label{eq:OmegappEMD}\\
     \ii\Sig\Om_{,pq}
     -(q+\ii p)(\Om_{,q}+\ii\Om_{,p})
     &=
     -\frac{\ii\Om\Sig}{\RR}B_-\bars{\Bal}
     +\frac{\ii\epss\Om\Sig}{\alphao^2}C_pC_q .
     \label{eq:OmegamixedEMD}
    \end{align}
\end{subequations}
}
\end{proposition}

The derivation from the relevant components of \eqref{eq:EinsteinEMD} is presented in Appendi \cref{app:OmegaDerivation}.  It is already instructive to interpret these relations at this point.  The first two equations specify the curvature of $\Om$ along each principal coordinate: the misaligned electromagnetic 1-form $B_-$ and the 1-form $C$ act exactly as the sources for these second derivatives.  The third equation links the two principal directions and constrains their mutual consistency.

On the electrovacuum boundary, where $C=0$, these equations reduce to the standard Einstein–Maxwell reconstruction.  When the Maxwell field is aligned as well, so that $B_-=0$, the diagonal components immediately imply \(\Om_{,pp}=\Om_{,qq}=0\), consequently, $\Om$ must be an affine function of each coordinate separately, \(\Om=c_0+c_1p+c_2q+c_3pq\). Inserting this expression into the mixed equation \eqref{eq:OmegamixedEMD} the linear coefficients to vanish, giving $c_1=c_2=0$.  After performing the conventional normalization of the constant term and the $pq$ coefficient, one obtains
\[
 \Om=1-\acc pq,
\]
which is exactly the Plebański–Demiański conformal factor.  In the EMD case, however, the terms involving $C_p$ and $C_q$ are not required to vanish, so a nontrivial scalar field can, in principle, sustain more general conformal factors.

The reconstruction equations are only necessary conditions.  Any candidate obtained from them must still satisfy the full potential system \eqref{eq:Aeq}-\eqref{eq:Ceq}.  This distinction will be important below: the Einstein reconstruction restricts the admissible geometry, while the potential equations decide whether that geometry is supported by an actual EMD solution.

\subsection{Product-separable reconstruction}
\label{sec:product}

Before fixing the Pleba\'nski--Demia\'nski form of $\Om$, it is useful to consider a broader separable class.  We restrict here to the ordinary dilatonic sign $\epss=+1$ and assume locally
\[
     \Om=\Om_p(p)\Om_q(q),
     \qquad
     C=C_p(p)\dd p+C_q(q)\dd q.
\]
The subscripts on $\Om_p$ and $\Om_q$ denote their arguments and not differentiation.  The same is true for the scalar components: each depends only on its corresponding principal coordinate. Consequently, $\dd C=0$ is automatic, as required locally for the scalar field.

\begin{proposition}[Product reduction]
In a connected region where
\[
     \PP>0,\qquad \QQ>0,\qquad \RR>0,
\]
the two diagonal reconstruction equations are equivalent to the existence of a real constant $\ell^2\geq0$ such that
\[
     \PP\left(
     \frac{\Om_p''}{\Om_p}-\frac{C_p^2}{\alphao^2}
     \right)
     =-\ell^2,
     \qquad
     \QQ\left(
     \frac{\Om_q''}{\Om_q}-\frac{C_q^2}{\alphao^2}
     \right)
     =\ell^2,
\]
and
\[
     |B_-|^2=\frac{\ell^2\RR}{\PP\QQ}.
\]
Once these relations hold, the mixed reconstruction equation fixes the product between the two relevant complex 1-forms as
\[
     B_-\bars{\Bal}
     =
     \frac{\ii\RR}{\Sig}
     \left[
     \ii\Sig\frac{\Om_{,pq}}{\Om}
     -(q+\ii p)
     \frac{\Om_{,q}+\ii\Om_{,p}}{\Om}
     -\frac{\ii\Sig}{\alphao^2}C_pC_q
     \right].
\]
\end{proposition}

\begin{proof}
When $\Om$ is expressed in product form, dividing the two diagonal reconstruction equations by $\Om$ yields
\[
     \frac{\Om_q''}{\Om_q}-\frac{C_q^2}{\alphao^2}
     =
     \frac{\PP}{\RR}|B_-|^2,
     \qquad
     \frac{\Om_p''}{\Om_p}-\frac{C_p^2}{\alphao^2}
     =
     -\frac{\QQ}{\RR}|B_-|^2.
\]
Multiplying the first equation by $\QQ(q)$ and the second by $\PP(p)$ gives
\[
     \QQ
     \left(
     \frac{\Om_q''}{\Om_q}-\frac{C_q^2}{\alphao^2}
     \right)
     =
     \frac{\PP\QQ}{\RR}|B_-|^2,
\]
and
\[
     \PP
     \left(
     \frac{\Om_p''}{\Om_p}-\frac{C_p^2}{\alphao^2}
     \right)
     =
     -\frac{\PP\QQ}{\RR}|B_-|^2.
\]
In the first equation, the left-hand side is a function of $q$ alone, while in the second it is a function of $p$ alone, and on a connected two-dimensional domain they are everywhere negatives of each other. It follows that both sides must be constant. Furthermore, since $\PP$, $\QQ$, and $\RR$ are positive, the shared term $\PP\QQ|B_-|^2/\RR$ cannot be negative. We therefore write it as $\ell^2\geq0$, yielding the three relations stated earlier. Lastly, by solving the mixed reconstruction equation for $B_-\bars{\Bal}$, one obtains the final formula.
\end{proof}

The parameter $\ell$ thus quantifies the size of the non-aligned electromagnetic sector in this product-separable reconstruction: setting $\ell=0$ is the same as requiring $B_-=0$ on the region of interest. This simplification, however, does not fix the relative complex phase of the electromagnetic 1-forms. In particular, the $C$ equation continues to provide a single real constraint involving $B_p$ and $B_q$, while the complex equations for $A$ and $B$ still need to be enforced separately. The proposition should therefore be interpreted as a simplification of the Einstein reconstruction rather than a full solution to the EMD system.

\subsection{Dilatonic rigidity in a physical Carter region}
\label{sec:rigidity}

We now restrict attention to the region where the sign assumption supporting the rigidity results is well defined. We refer to a connected domain as a \emph{physical Carter region} when
\begin{equation}
     \Om>0,\quad
     \PP>0,\quad
     \QQ>0,\quad
     \RR=\QQ-\PP>0,\quad
     \epss=+1.
     \label{eq:physicalpatch}
\end{equation}
The conditions $\PP>0$ and $\QQ>0$ fix the Riemannian signature of the two-dimensional Carter base, while $\RR>0$ implies $f>0$ in the stationary region.  The choice $\epss=+1$ is equally important: it makes the scalar contribution to the diagonal reconstruction equations positive.  All statements below are local within such a connected region.

\subsubsection{The bilinear PD obstruction}

The first rigidity result follows directly from the special curvature properties of the bilinear PD conformal factor.

\begin{theorem}[Bilinear obstruction]\label{thm:bilinearObstruction}
Let \(\Om=1-\acc pq \),  in a physical Carter region.  Then every local EMD solution with $\alphao\neq0$ within this conformal Carter ansatz satisfies \( C=0, \quad B_-=0\), hence the dilatonic scalar is constant and the Maxwell field is aligned with the two Carter principal directions.
\end{theorem}

\begin{proof}
For the bilinear factor $\Om=1-\acc pq$, \(\Om_{,pp}=\Om_{,qq}=0\), thnen the $qq$ reconstruction equation therefore becomes \( 0  =  \frac{\Om\PP}{\RR}|B_-|^2  +\frac{\Om}{\alphao^2}C_q^2\).  Every factor multiplying the two squares is positive in the physical region \eqref{eq:physicalpatch}.  Both terms on the right-hand side are consequently nonnegative, and their sum can vanish only if \( B_-=0,  \quad  C_q=0\). Substituting $B_-=0$ into the $pp$ equation gives \( 0=\frac{\Om}{\alphao^2}C_p^2\) and therefore $C_p=0$.  Hence $C=0$ identically on the connected region.
\end{proof}

The obstruction is thus a direct consequence of positivity.  The PD factor has vanishing pure second derivatives, whereas an ordinary scalar gradient and a non-aligned Maxwell field would contribute nonnegative sources to precisely those derivatives.  They cannot be present without changing the bilinear form of $\Om$.

The theorem still allows an aligned electromagnetic 1-form.  At nonzero dilatonic coupling, however, the scalar equation imposes an additional restriction.

\begin{corollary}[PD Coulombic-type 1-form at nonzero dilatonic coupling]
Under the hypotheses of the preceding theorem, the scalar equation requires \(\Rea(B\cdot B)=0\).  If $B$ is the aligned PD Coulombic-typw 1-form \eqref{eq:PDBcurrent} on an open two-dimensional region, this condition implies \(e=g=0\). 
\end{corollary}

\begin{proof}
Once the theorem has established $C=0$, the scalar equation \eqref{eq:Ceq} reduces to $\Rea(B\cdot B)=0$.  For the PD Coulombic 1-form we may write \(B_p=\ii b,  \quad B_q=b,  \quad  b=\frac{e+\ii g}{\sqrt f\,(q+\ii p)^2}\). Using \eqref{eq:Ceq} together with $f=\RR/(\Om^2\Sig)$ gives \(B\cdot B = \frac{\Om^2\Sig}{\lambda} \frac{(e+\ii g)^2}{(q+\ii p)^4}\). The prefactor $\Om^2\Sig/\lambda$ is real and nonzero in the region under consideration.  Thus $\Rea(B\cdot B)=0$ is equivalent, after multiplication by the nonvanishing real denominator, to the polynomial identity
\[
     (e^2-g^2)(p^4-6p^2q^2+q^4)
     -8eg\,pq(p^2-q^2)=0.
\]
Because this identity holds on an open subset of the $(p,q)$ plane, every polynomial coefficient must vanish.  The coefficient of $p^4$ gives $e^2=g^2$, while that of $p^3q$ gives $eg=0$.  The two conditions are compatible only for \( e=g=0\).
\end{proof}

Thus the obstruction has two logically distinct steps.  The Einstein reconstruction first forces the scalar 1-form to vanish and aligns the Maxwell field.  Once $C=0$, the scalar field equation remains nontrivial because $\alphao\neq0$ and excludes the nonzero PD Coulombic charge itself.

\begin{remark}
The hypotheses of the theorem are essential.  In particular, the positivity argument does not extend directly to a phantom scalar, $\epss=-1$, because the scalar term then enters the reconstruction equations with the opposite sign and may cancel the electromagnetic contribution.  The argument may also fail in regions where the signs of $\PP/\RR$ or $\QQ/\RR$ change, or when a scalar potential modifies the field equations.  Most importantly, the theorem concerns the bilinear factor $\Om=1-\acc pq$ and does not exclude dilatonic conformal Carter geometries with a different functional form of $\Om$.
\end{remark}
%

\subsubsection{Exact quadratic St\"ackel 1-forms}

The bilinear theorem naturally raises a second question: does the obstruction persist if the conformal factor is allowed to depart from the exact PD form?  To examine this without abandoning the Carter structure, we consider the lowest-degree polynomial extension containing both the bilinear PD term and the biquadratic combinations that naturally arise in conformal Carter geometries.  We therefore assume $\acc>0$ and write
\begin{equation}
     U:=\Om^2
     =
     1-2\acc pq+\zeta(q^2-p^2)
     +\mu p^2q+\nu pq^2+\delta p^2q^2 .
     \label{eq:Uansatz}
\end{equation}
For the dilatonic 1-form we take the corresponding quadratic St\"ackel form
{\footnotesize
\begin{equation}
     \frac{C}{\alphao}
     =
     \frac{v(q)}{U}\dd p+\frac{u(p)}{U}\dd q,
     \quad
     u=u_0+u_1p+u_2p^2,
     \quad
     v=v_0+v_1q+v_2q^2 .
     \label{eq:CStackel}
\end{equation}
}
This mixed dependence is deliberate: aside from the shared factor $U^{-1}$, the coefficient multiplying $\dd p$ is a function only of $q$, whereas the coefficient multiplying $\dd q$ is a function only of $p$. Because $C$ is locally proportional to the gradient of the dilaton field, it must likewise be exact. As we will demonstrate, this requirement already eliminates most of the polynomial deformation even before any field equation is applied.

\begin{lemma}[Exactness classification]\label{lem:exactnessClassification}
Let $U,u,v$ be real and given by \eqref{eq:Uansatz}-\eqref{eq:CStackel}.  If $C$ is exact and nonzero on an open set, then necessarily
\begin{equation}
     \mu=\nu=\zeta=0,
     \qquad
     U=1-2\acc pq+\delta p^2q^2 .
     \label{eq:Uclassified}
\end{equation}
For $\delta\neq\acc^2$, the 1-fomr reduces to
\[
     u=hp,
     \qquad
     v=hq,
     \qquad
     \frac{C}{\alphao}
     =
     h\,\frac{\dd(pq)}{U}.
\]
At the limiting value $\delta=\acc^2$, exactness admits instead the larger family
\[
     u=u_0+hp+\acc v_0p^2,
     \qquad
     v=v_0+hq+\acc u_0q^2.
\]
\end{lemma}

\begin{proof}
Because $\acc>0$, we begin by rescaling $p$ and $q$ via $\bar p=\sqrt{\acc}\,p$ and $\bar q=\sqrt{\acc}\,q$. With the induced redefinition of the coefficients, the algebraic classification can be reduced to the representative case $\acc=1$, after which the original scaling may be reinstated at the end.

For the one-form \eqref{eq:CStackel}, the condition of being exact is $\partial_p(u/U)=\partial_q(v/U)$.  Clearing the denominator by multiplying through by $U^2$ yields the polynomial identity
\[
     u'U-uU_{,p}=v'U-vU_{,q},
\]
where $u'=\dd u/\dd p$ and $v'=\dd v/\dd q$.  Since this equality is required to hold throughout an open set in the $(p,q)$ plane, the coefficients of each linearly independent monomial must match.  In particular, this implies $u_1=v_1=:h$, and furthermore
{\footnotesize
\[
\begin{aligned}
     &\nu u_2=\mu v_2,
     &&2\zeta h+\mu v_0=0,
     &&2\zeta h-\nu u_0=0,
     &&\mu u_0=\nu v_0,\\
     &u_2=v_0-\zeta u_0,
     &&v_2=u_0+\zeta v_0,
\end{aligned}
\]
}
together with the remaining two conditions coming from the highest-order coefficients.

These relations can be analyzed case by case. If $(\mu,\nu)\neq(0,0)$ and $\zeta\neq0$, then the subsystem in $(u_0,v_0,h)$ has determinant $2\zeta(\mu^2+\nu^2)\neq0$, so it follows that $u_0=v_0=h=0$; substituting this back into the remaining equations then yields $u_2=v_2=0$. If instead $\zeta=0$ with $(\mu,\nu)\neq(0,0)$, the same result is obtained from the conditions $\mu v_0=0$, $\nu u_0=0$, and $\mu u_0=\nu v_0$. Therefore, an exact 1-form can be nonzero only if \( \mu=\nu=0\).

Assuming $\mu=\nu=0$, one must also rule out $\zeta\neq0$. In fact, the remaining coefficient relations give
\[
     \left[4\zeta^2+(\zeta^2+\delta-1)^2\right]
     (u_0^2+v_0^2)=0.
\]
For real $\zeta\neq0$, the bracketed factor is strictly positive, which forces $u_0=v_0=0$, and then the remaining coefficient equations again imply that the entire 1-fomr is zero. Hence, a nonzero exact 1-fomr additionally requires $\zeta=0$.

Thus, in the normalized variables the only remaining conformal factor is $U=1-2pq+\delta p^2q^2$. Exactness then imposes $u_1=v_1=h$, $u_2=v_0$, and $v_2=u_0$, along with the constraints $(\delta-1)u_0=(\delta-1)v_0=0$. When $\delta\neq1$, this forces $u=hp$ and $v=hq$ as the sole possibility, while for $\delta=1$ the constants $u_0$ and $v_0$ are still undetermined. Reintroducing the original scale $\acc$ yields the two branches stated earlier.
\end{proof}

This lemma is entirely kinematic: it does not rely on any Einstein or Maxwell equations. Exactness by itself compels the quadratic St\"ackel 1-form to remove the $\mu$, $\nu$, and $\zeta$ deformations, leaving only a conformal factor that depends solely on the product $pq$. We may now apply the Einstein reconstruction and check whether the remaining scalar 1-form is compatible with a non-aligned Maxwell field.

\begin{theorem}[Polynomial dilatonic rigidity]\label{thm:polynomialRigidity}
Assume the physical region \eqref{eq:physicalpatch}, the polynomial conformal factor \eqref{eq:Uansatz}, and a nonzero exact scalar 1-form of the form \eqref{eq:CStackel}.  Then the Maxwell field must be aligned, \( B_-=0\).

More precisely, the only nonzero scalar branch admitted by the Einstein reconstruction equations is
{\footnotesize
\begin{equation}
     U=1-2\acc pq+\delta p^2q^2,
     \qquad
     \frac{C}{\alphao}
     =
     h\,\frac{\dd(pq)}{U},
     \qquad
     h^2=\delta-\acc^2>0.
     \label{eq:hdelta}
\end{equation}
}
The limiting case $\delta=\acc^2$ reduces to the bilinear PD factor and is incompatible with a nonzero scalar 1-form.
\end{theorem}

\begin{proof}
By \cref{lem:exactnessClassification}, it suffices to treat the case \(U=1-2\acc pq+\delta p^2q^2\). Assume first that $\delta\neq\acc^2$ and set the single variable $s:=pq$.  Then
\[
     \Om=\sqrt{1-2\acc s+\delta s^2}, \quad  \frac{C}{\alphao}=h\,\frac{\dd s}{U},
\]
so the conformal factor and the scalar field are functions of $s$ alone.  Using $\partial_q s=p$ and $\partial_p s=q$, the two diagonal reconstruction conditions become
{\footnotesize
\[
     p^2\left(
     \frac{\Om''}{\Om}-\frac{h^2}{U^2}
     \right)
     =
     \frac{\PP}{\RR}|B_-|^2,
     \qquad
     q^2\left(
     \frac{\Om''}{\Om}-\frac{h^2}{U^2}
     \right)
     =
     -\frac{\QQ}{\RR}|B_-|^2,
\]
}
where primes indicate differentiation with respect to $s$.

The left-hand sides coincide up to the positive factors $p^2$ and $q^2$, while in the physical Carter region the electromagnetic contributions on the rigth-hand sides carry opposite signs.  Hence, on the open set where $pq\neq0$ (so $p^2,q^2>0$), compatibility force
\[
     B_-=0,
     \qquad
     \frac{\Om''}{\Om}=\frac{h^2}{U^2}.
\]
Continuity then propagates the alignment condition across the entire connected region.

For $\Om=\sqrt{1-2\acc s+\delta s^2}$ one computes directly that
\[
     \frac{\Om''}{\Om}
     =
     \frac{\delta-\acc^2}{U^2}.
\]
Comparing with the previous identity gives $h^2=\delta-\acc^2$.  Since the scalar 1-form is taken to be real and nonvanishing, we require $h^2>0$, and therefore $\delta>\acc^2$.

If instead $\delta=\acc^2$, then $U=(1-\acc pq)^2$.  Because $\Om>0$ in the physical region, this implies $\Om=1-\acc pq$, i.e. the standard bilinear PD factor.  The earlier theorem would then impose $C=0$, contradicting the assumption that the scalar 1-form is nonzero.  Thus the borderline case cannot occur.
\end{proof}



\subsubsection{Full closure of the aligned branch}

The preceding theorem exhausts the Einstein reconstruction within \eqref{eq:Uansatz}-\eqref{eq:CStackel}. It remains to test the unique aligned candidate \eqref{eq:hdelta} against the potential equations.

\begin{theorem}[Quadratic St\"ackel no-go]
\label{thm:quadraticStackelNoGo}
Let $\alphao\neq0$ and $\acc>0$. In a physical Carter region, no smooth structure functions $\PP(p)$ and $\QQ(q)$ can make the ans\"atze \eqref{eq:Uansatz} and \eqref{eq:CStackel}, with a nonzero exact scalar 1-form, solve the full EMD system.
\end{theorem}

\begin{proof}
By \cref{lem:exactnessClassification,thm:polynomialRigidity}, any such solution must lie on the branch \eqref{eq:hdelta} and satisfy $B_-=0$. Let $B_p=\ii b$ and $B_q=b$. The definitions \eqref{eq:kappaC} and \eqref{eq:ABC} imply identically $\dd B=-\Rea(A)\wedge B-C\wedge\bars B$. Combining this identity with
\eqref{eq:Beq} and \eqref{eq:AConformalSplit} gives
\[
\begin{aligned}
     \partial_p b&=
     \left[-\Rea(A_p)-\frac{2\ii}{q+\ii p}\right]b-C_p\bars b,\\
     \partial_q b&=
     \left[-\Rea(A_q)-\frac{2}{q+\ii p}\right]b+C_q\bars b .
\end{aligned}
\]
Since $\Rea(A)=\tfrac12\dd\ln f$ and $\dd C=0$, commutation of the two derivatives yields
\[
     \left[
     -2\partial_qC_p+\frac{4\ii}{\Sig}(pC_p-qC_q)
     \right]\bars b=0.
\]
For \eqref{eq:hdelta}, $C_p=\alphao hq/U$ and $C_q=\alphao hp/U$.
Thus $pC_p-qC_q=0$, whereas $\partial_qC_p=\alphao h(1-\delta p^2q^2)/U^2$. Since $h\neq0$, this coefficient cannot vanish on a two-dimensional open set. Hence $b=0$ off a set with empty interior, and smoothness gives $B=0$ throughout the patch.

With $B=0$, \eqref{eq:Ceq} reduces to
\[
     U(q\PP'+p\QQ')
     -4(\delta pq-\acc)(q^2\PP+p^2\QQ)=0.
\]
On a subpatch with $pq\neq0$, set $x=\ln|p|$, $y=\ln|q|$, $r=x+y$, $w=x-y$, and $H=\PP(p)/p^2+\QQ(q)/q^2$. The preceding equation becomes $UH_{,r}+(1-\delta s^2)H=0$, with $s=pq$, and therefore $H=(U/s)F(w)$. Since $H_{,xy}=0$ by construction,
\[
 0=H_{,xy}=\frac{U}{s}(F-F'')+2\acc F .
\]
Because $\acc>0$ and $U/s=s^{-1}-2\acc+\delta s$ is nonconstant, $F=0$. Consequently, $\PP=cp^2$ and $\QQ=-cq^2$, contradicting $\PP>0$ and $\QQ>0$. Hence no nonzero scalar branch solves the full potential system.
\end{proof}

The theorem does not assume that $\PP$ or $\QQ$ are polynomial, however, it is a local result and relies on $\epss=+1$, $\acc>0$, the polynomial conformal factor \eqref{eq:Uansatz}, and the precise quadratic St\"ackel structure \eqref{eq:CStackel}. As a consequence, phantom scalars, conformal factors of higher degree, and non-St\"ackel currents are not covered by this theorem.

\subsection{Known dilatonic boundary and scope of the rigidity results}
\label{sec:boundaries}

The results quoted above are conditional rigidity statements and should not be taken to mean that dilatonic conformal Carter geometries are universally ruled out. A helpful boundary-case example is already available in the uncharged sector $B=0$, where one local branch is given by
\[
     \PP=a_0+a_2p^2+\beta b_0p^4,
     \qquad
     \QQ=b_0-a_2q^2+\beta a_0q^4,
\]
together with
\[
 \Om=\frac{1+\beta p^2q^2}{\sqrt{pq}},
 \qquad
 \phi=\phi_0\mp\frac{\sqrt3}{2}\ln(pq),
 \qquad pq>0.
\]
The corresponding dilatonic 1-form is
\[
 C
 =
 \mp\frac{\sqrt3\alphao}{2}
 \left(
 \frac{\dd p}{p}
 +
 \frac{\dd q}{q}
 \right).
\]
Following a coordinate inversion, this corresponds to the $\nu=2$ zero-potential, minimally coupled scalar subbranch examined by Anabal\'on \cite{Anabalon:2012ta}. We present it as a boundary check for our classification, not as a novel solution.

Its significance is especially evident because its conformal factor falls outside both of the categories discussed above: it is neither the bilinear PD factor nor an element of the polynomial class \eqref{eq:Uansatz}. Hence it evades any conflict with either rigidity theorem. The first theorem rules out a nonconstant ordinary dilaton when the conformal factor is fixed precisely to $\Om=1-\acc pq$; the second establishes that, in the quadratic St\"ackel closure, any remaining nonzero scalar 1-form necessarily enforces electromagnetic alignment. Dilatonic Carter configurations with an alternative functional form of $\Om$ therefore lie beyond the reach of these results.

\section{Conclusions}

Within a generalized Ernst potential-space framework, we have assembled a complete local dictionary for the \(\Lambda=0\) Plebański-Demiański electrovacuum. This dictionary provides the Weyl coordinates and metric functions, the electric, magnetic, and twist potentials, the target-space 1-forms \(A\), \(B\), and \(C\), and the Hamiltonian tensor used to recover the remaining Weyl function \(k\).

Through the principal chiral decomposition, the dictionary acquires a clear geometric meaning. One chirality of \(B\) quantifies the obstruction to double electromagnetic alignment, whereas writing \(A=A_{\rm C}-\dd\ln\Om\) cleanly splits the Carter seed from the conformal deformation. For the normalized PD factor, the opposite chirality of \(A\) is exactly zero in the nonaccelerating limit. In addition, the trace-free Hamiltonian tensor captures the relative conical information associated with disconnected portions of the symmetry axis.

For a standard scalar field in a physical Carter region, the bilinear term \(\Om=1-\acc pq\) enforces \(C=0\) and \(B_-=0\) without needing to assume polynomial Carter functions or to impose alignment in advance. With a fixed nonzero dilatonic coupling, the remaining scalar equation further rules out any nonzero PD Coulomb-type current on a nondegenerate cohomogeneity-two patch.

A quadratic Stäckel analysis sharpens these results. Imposing exactness restricts the polynomial conformal factor to \(U=1-2\acc pq+\delta p^2q^2\), and the Einstein reconstruction leaves only the aligned option \(C/\alphao=h\,\dd(pq)/U\), where \(h^2=\delta-\acc^2>0\). Consistency with the \(B\)-field equation then requires its Maxwell-type contribution to vanish, and the scalar equation subsequently yields \(\PP=cp^2\) and \(\QQ=-cq^2\), which conflicts with the sign requirements defining a physical Carter region. Therefore, within this quadratic Stäckel class, there is no local nonconstant ordinary-dilaton solution, regardless of any polynomial assumption about \(\PP\) and \(\QQ\).

This result is local and conditional—a no-go statement rather than a global no-hair theorem. Cases involving phantom scalars, Carter regions with sign changes, scalar potentials, conformal factors of higher degree, or genuinely non-Stäckel scalar currents are not covered.

\section{Acknowledgements}

LB thanks SECIHTI-M\'exico for the doctoral grant.
This work was also partially supported by SECIHTI M\'exico under grants SECIHTI CBF-2025-G-1720 and CBF-2025-G-176.

\section*{Data Availability Statement}
This article presents entirely theoretical research. No datasets were generated or analyzed during the 1-form study.

\begin{appendices}

\section{Component proof of the NP--potential dictionary}
\label{app:NPproof}

The EMD magnetic-potential definition gives
\begin{align}
     A_{\varphi,p}+\omega A_{t,p}
     &=\frac{\Sig\QQ}{2\kappa^2\RR}\chi_{,q},
     \label{eq:Avarphip}\\
     A_{\varphi,q}+\omega A_{t,q}
     &=-\frac{\Sig\PP}{2\kappa^2\RR}\chi_{,p}.
     \label{eq:Avarphiq}
\end{align}
Resolving $F=\dd\mathcal A$ in the coframe \eqref{eq:coframe} yields
\begin{align}
     F_{01}&=-\frac{\Om^2}{2\RR}
     \left(\PP\kappa^{-2}\chi_{,p}+\QQ\psi_{,q}\right),
     \label{eq:F01}\\
     F_{02}&=\frac{\Om^2\sqrt{\PP\QQ}}{2\RR}
     \left(\kappa^{-2}\chi_{,q}-\psi_{,p}\right),
     \label{eq:F02}\\
     F_{13}&=-\frac{\Om^2\sqrt{\PP\QQ}}{2\RR}
     \left(\kappa^{-2}\chi_{,p}+\psi_{,q}\right),
     \label{eq:F13}\\
     F_{23}&=\frac{\Om^2}{2\RR}
     \left(\QQ\kappa^{-2}\chi_{,q}-\PP\psi_{,p}\right).
     \label{eq:F23}
\end{align}
The circularity components $F_{03}$ and $F_{12}$ vanish.  Hence
{\footnotesize
\begin{align}
     \Phi_0=\Phi_2
     &=-\frac{\Om^2\sqrt{\PP\QQ}}{4\RR}
     \left[
     \kappa^{-2}\chi_{,p}+\psi_{,q}
     +\ii(\psi_{,p}-\kappa^{-2}\chi_{,q})
     \right],
     \label{eq:Phi0pot}\\
     \Phi_1
     &=\frac{\Om^2}{4\RR}
     \left[
     \PP\kappa^{-2}\chi_{,p}+\QQ\psi_{,q}
     +\ii(\PP\psi_{,p}-\QQ\kappa^{-2}\chi_{,q})
     \right].
     \label{eq:Phi1pot}
\end{align}
}
On the other hand, \eqref{eq:ABC} gives
{\footnotesize
\begin{align}
     \frac{\ii}{\kappa}(B_p-\ii B_q)
     &=-\frac1{2\sqrt f}
     \left[
     \kappa^{-2}\chi_{,p}+\psi_{,q}
     +\ii(\psi_{,p}-\kappa^{-2}\chi_{,q})
     \right],
     \label{eq:Bminuspot}\\
     \frac1\kappa(\QQ B_q+\ii\PP B_p)
     &=-\frac1{2\sqrt f}
     \left[
     \QQ\psi_{,q}+\PP\kappa^{-2}\chi_{,p}
     +\ii(\PP\psi_{,p}-\QQ\kappa^{-2}\chi_{,q})
     \right].
     \label{eq:Bpluspot}
\end{align}
}
Equations \eqref{eq:Phi0pot}-\eqref{eq:Bpluspot} prove \eqref{eq:Phi0B}-\eqref{eq:Phi1B} without using any field equation.

\section{Derivation of the EMD conformal-factor equations}
\label{app:OmegaDerivation}

In the principal tetrad, the three geometric Ricci-spinor projections relevant for the conformal factor can be normalized as
\begin{align}
     \mathfrak R_{00}
     &=\frac{\QQ\Om}{2\Sig}\Om_{,qq},
     \label{eq:R00geom}\\
     \mathfrak R_{02}
     &=-\frac{\PP\Om}{2\Sig}\Om_{,pp},
     \label{eq:R02geom}\\
     \mathfrak R_{01}
     &=\frac{\sqrt{\PP\QQ}\Om}{2\Sig^2}
     \left[
     \ii\Sig\Om_{,pq}-(q+\ii p)(\Om_{,q}+\ii\Om_{,p})
     \right].
     \label{eq:R01geom}
\end{align}
For the Maxwell part, the convention is $\mathfrak R_{ij}=2\kappa^2\Phi_i\bars\Phi_j$.  The scalar is stationary and axisymmetric, so its tetrad derivatives contain only $C_p/\alphao$ and $C_q/\alphao$.  Adding $2\epss\phi_{,\mu}\phi_{,\nu}$ to the three projections gives
\begin{align}
     \mathfrak R_{00}
     &=2\kappa^2\Phi_0\bars\Phi_0
     +\frac{\epss\Om^2\QQ}{2\alphao^2\Sig}C_q^2,
     \label{eq:R00matter}\\
     \mathfrak R_{02}
     &=2\kappa^2\Phi_0\bars\Phi_2
     -\frac{\epss\Om^2\PP}{2\alphao^2\Sig}C_p^2,
     \label{eq:R02matter}\\
     \mathfrak R_{01}
     &=2\kappa^2\Phi_0\bars\Phi_1
     +\frac{\ii\epss\Om^2\sqrt{\PP\QQ}}{2\alphao^2\Sig}C_pC_q.
     \label{eq:R01matter}
\end{align}
Substituting \eqref{eq:Phi0B}-\eqref{eq:Phi1B}, using $f=\RR/(\Om^2\Sig)$, and comparing \eqref{eq:R00geom}-\eqref{eq:R01geom} with \eqref{eq:R00matter}-\eqref{eq:R01matter} gives \eqref{eq:OmegaqqEMD}--\eqref{eq:OmegamixedEMD}.  This also displays the cancellation of $\kappa$ in the Einstein equations.

\section{Foundational calculus, Routh reduction, and the roots of the Hamiltonian framework}
\label{app:formalism}

This appendix lays out the procedural steps underlying \cref{sec:potential}. They are written out in full so that the sign choices in the PD dictionary can be verified without relying on conventions adopted from another paper.

\subsection{The operators on the Weyl base}

Given a function $u(\rho,z)$, set $Du=(u_{,\rho},u_{,z})$ and $\widetilde Du=(u_{,z},-u_{,\rho})$. Using the standard Euclidean inner product, one has $Du\cdot Dv=u_{,\rho}v_{,\rho}+u_{,z}v_{,z}$. A direct computation shows $\widetilde D^2u=-Du$, $\widetilde Du\cdot\widetilde Dv=Du\cdot Dv$, $Du\cdot\widetilde Dv=-Dv\cdot\widetilde Du$, and $Du\cdot\widetilde Du=0$. For a 1-form $X=X_\rho\dd\rho+X_z\dd z$,
{\footnotesize
\begin{equation}
     D(\rho X)=\partial_\rho(\rho X_\rho)+\partial_z(\rho X_z),
     \quad
     \frac1\rho D(\rho\dd u)=u_{,\rho\rho}+\rho^{-1}u_{,\rho}+u_{,zz}.
     \label{eq:weightedBaseOperators}
\end{equation}
}
With the Hodge convention adopted in the main text, $(\widetilde Du)^\flat=-\star_{\Base}\dd u$. Therefore, if $s$ is an axisymmetric harmonic function, meaning $\rho^{-1}D(\rho\dd s)=0$, then there exists a local dual potential $\widetilde s$ such that
\begin{equation}
 \dd\widetilde s=\rho\star_{\Base}\dd s,
 \qquad
 \dd^2\widetilde s=\dd(\rho\star_{\Base}\dd s)=0.
 \label{eq:harmonicDual}
\end{equation}
In particular, this yields the desired integrability result without the need to posit a conjugate coordinate a priori.

\subsection{From the four-dimensional action to the target metric}

Substituting \eqref{eq:WeylMetric} and $\mathcal A=A_t\dd t+A_\varphi\dd\varphi$ into \eqref{eq:action}, integrating out the Killing directions, discarding the total-derivative term in the gravitational part, and keeping the flat-base contractions introduced above.  Prior to defining the magnetic and twist potentials, the reduced Lagrangian density reads
\begin{align}
     \mathcal L_{\rm eff}={}&\frac{\rho}{2f^2}(Df)^2
     -\frac{f^2}{2\rho}(D\omega)^2
     +\frac{2\rho\epss}{\alphao^2\kappa^2}(D\kappa)^2\nonumber\\
     &+\frac{2f\kappa^2}{\rho}\left[
     (\omega D A_t+D A_\varphi)^2
     -\frac{\rho^2}{f^2}(D A_t)^2\right].
     \label{eq:LeffBeforeDualization}
\end{align}
Varying with respect to $A_\varphi$ and $\omega$ yields two conserved currents.  Their local first integrals coincide with the definitions of $\chi$ and $\epsilon$ given in \eqref{eq:potdefinitions}.  Carrying out the associated Routh transform, rather than simply substituting algebraically into \eqref{eq:LeffBeforeDualization}, one obtains
\begin{align}
     \mathcal L_{\rm pot}=&\rho \bigg[
     \frac{(Df)^2+(D\epsilon-\psi D\chi)^2}{2f^2}
     +\frac{2\epss}{\alphao^2\kappa^2}(D\kappa)^2 \nonumber
     \\
     & -\frac{\kappa^2}{2f}(D\psi)^2
     -\frac{1}{2f\kappa^2}(D\chi)^2 \bigg].
     \label{eq:LpotDerived}
\end{align}
This equals $\rho$ times the pullback of \eqref{eq:targetmetric}.  The corresponding Euler-Lagrange equations are \eqref{eq:harmonicMapEquation}, writing their five real components in the coframe \eqref{eq:ABC} reproduces \eqref{eq:Aeq}--\eqref{eq:Ceq}.  Notably, both the relative minus sign in the electromagnetic sector and the $(\kappa,\kappa^{-1})$ factors appearing in $B$ arise from the Routh transformation.

\subsection{Hamiltonian and its polarization}

Restrict the target map to a curve $Y^A(\iota)$ and put $L_{\Target}=\tfrac12G_{AB}\dot Y^A\dot Y^B$.  Explicitly,
\begin{equation}
     L_{\Target}=\frac{\dot f^2+(\dot\epsilon-\psi\dot\chi)^2}{4f^2}
     +\frac{\epss\dot\kappa^2}{\alphao^2\kappa^2}
     -\frac{\kappa^2\dot\psi^2}{4f}
     -\frac{\dot\chi^2}{4f\kappa^2}.
     \label{eq:targetCurveLagrangian}
\end{equation}
The five conjugate momenta are
{\footnotesize
\begin{align}
     p_f=\frac{\dot f}{2f^2},\quad
     p_\epsilon=\frac{\dot\epsilon-\psi\dot\chi}{2f^2},\quad
     p_\psi=-\frac{\kappa^2\dot\psi}{2f},\quad \nonumber
     \\
     p_\kappa=\frac{2\epss\dot\kappa}{\alphao^2\kappa^2},\quad
     p_\chi=-\psi p_\epsilon-\frac{\dot\chi}{2f\kappa^2}.
     \label{eq:targetMomenta}
\end{align}
}
The Legendre transform is therefore
\begin{equation}
 H_{\Target}=f^2(p_f^2+p_\epsilon^2)
 -\frac f{\kappa^2}p_\psi^2
 -f\kappa^2(p_\chi+\psi p_\epsilon)^2
 +\frac{\alphao^2\kappa^2}{4\epss}p_\kappa^2.
 \label{eq:targetPhaseHamiltonian}
\end{equation}
Let $a,b,c$ be the components of the forms in \eqref{eq:ABC} along the curve. Substituting \eqref{eq:targetMomenta} yields $H_{\Target}=a\bars a-b\bars b+\epss c^2/\alphao^2$.
Polarizing this quadratic Hamiltonian in two independent base directions gives the tensor \eqref{eq:HamiltonianTensor}, hence the name \emph{Hamiltonian tensor}. See Ref. \cite{Bixano:2026duf}.

Since the target metric is independent of $\epsilon$ and $\chi$, $p_\epsilon$ and $p_\chi$ are conserved along any target geodesic, write their values as $q_N$ and $p_N$. For a rank-one map, the momentum definitions give $\Theta=2f^2q_N\dd s$ and $\dd\chi=-2f\kappa^2(p_N+\psi q_N)\dd s$. The last line of \eqref{eq:basicDictionary} then integrates (without additional field equations) to
\begin{equation}
 \omega=\omega_0-2q_N\widetilde s,
 \qquad
 A_\varphi=A_{\varphi0}-\frac12\omega\psi+p_N\widetilde s.
 \label{eq:rankOneNoetherReconstruction}
\end{equation}
Differentiating and using \eqref{eq:harmonicDual} confirms both identities. These are the Noether charges associated with the cyclic target directions \cite{Bixano:2026duf}, for rank two they become conserved 1-forms.

If the full map has rank one, $Y^A=Y^A(s)$ with $s$ satisfying \eqref{eq:harmonicDual}, the field equations reduce to the affine geodesic equation and
\begin{equation}
 K_0=|a|^2-|b|^2+\frac{\epss}{\alphao^2}c^2
 =\frac12G_{AB}\frac{\dd Y^A}{\dd s}\frac{\dd Y^B}{\dd s}
 \label{eq:rankOneEnergyProof}
\end{equation}
is constant, since its derivative is the target velocity contracted with the geodesic equation.

For rank two, set $A_\zeta=A_\rho-\ii A_z$ and $\bars A_\zeta=\bars A_\rho-\ii\bars A_z$, and similarly for $B$ and $C$ (taking the bar before forming components). Expanding \eqref{eq:kHamiltonianConstraints} gives the Hopf-type quadrature
\begin{equation}
 k_{,\rho}-\ii k_{,z}=\rho\left(
 A_\zeta\bars A_\zeta-B_\zeta\bars B_\zeta
 +\frac{\epss}{\alphao^2}C_\zeta^2\right).
 \label{eq:HopfHamiltonian}
\end{equation}
Only for rank one does the parenthesis factor as $K_0(s_{,\rho}-\ii s_{,z})^2$. Since generic PD has rank two, replacing \eqref{eq:HamiltonianTensor} by a single constant Hamiltonian would discard information.

\end{appendices}

\clearpage
\bibliographystyle{elsarticle-harv} 
\bibliography{Bibliografia}

@article{Matos:2010pcd,
    author = "Matos, Tonatiuh",
    title = "{Class of Einstein-Maxwell Phantom Fields: Rotating and Magnetised Wormholes}",
    eprint = "0902.4439",
    archivePrefix = "arXiv",
    primaryClass = "gr-qc",
    reportNumber = "CIEA-09-GR5",
    doi = "10.1007/s10714-010-0976-6",
    journal = "Gen. Rel. Grav.",
    volume = "42",
    pages = "1969--1990",
    year = "2010"
}

@article{Matos:2000ai,
    author = "Matos, Tonatiuh and Nunez, Dario and Estevez, Gabino and Rios, Maribel",
    title = "{Rotating 5-D Kaluza-Klein space-times from invariant transformations}",
    eprint = "gr-qc/0001039",
    archivePrefix = "arXiv",
    reportNumber = "CINVESTAV-00-FIS-8",
    doi = "10.1023/A:1001982001694",
    journal = "Gen. Rel. Grav.",
    volume = "32",
    pages = "1499--1525",
    year = "2000"
}

@article{Matos:2000za,
    author = "Matos, Tonatiuh and Nunez, Dario and Rios, Maribel",
    title = "{Class of Einstein-Maxwell dilatons, an ansatz for new families of rotating solutions}",
    eprint = "gr-qc/0008068",
    archivePrefix = "arXiv",
    doi = "10.1088/0264-9381/17/18/323",
    journal = "Class. Quant. Grav.",
    volume = "17",
    pages = "3917--3934",
    year = "2000"
}

@article{Ernst:1967wx,
    author = "Ernst, Frederick J.",
    title = "{New formulation of the axially symmetric gravitational field problem}",
    doi = "10.1103/PhysRev.167.1175",
    journal = "Phys. Rev.",
    volume = "167",
    pages = "1175--1179",
    year = "1968"
}

@article{Ernst:1967by,
    author = "Ernst, Frederick J.",
    title = "{New Formulation of the Axially Symmetric Gravitational Field Problem. II}",
    doi = "10.1103/PhysRev.168.1415",
    journal = "Phys. Rev.",
    volume = "168",
    pages = "1415--1417",
    year = "1968"
}

@article{Plebanski:1976gy,
    author = "Plebanski, J. F. and Demianski, M.",
    title = "{Rotating, charged, and uniformly accelerating mass in general relativity}",
    doi = "10.1016/0003-4916(76)90240-2",
    journal = "Annals Phys.",
    volume = "98",
    pages = "98--127",
    year = "1976"
}

@article{Matos:1994hm,
    author = "Matos, Tonatiuh",
    title = "{5-D axisymmetric stationary solutions as harmonic maps}",
    eprint = "gr-qc/9401009",
    archivePrefix = "arXiv",
    doi = "10.1063/1.530590",
    journal = "J. Math. Phys.",
    volume = "35",
    pages = "1302--1321",
    year = "1994"
}

@article{Bixano:2026xum,
    author = "Bixano, Leonel and Matos, Tonatiuh",
    title = "{Generalized Ernst Potentials for arbitrary Dilatonic Theories}",
    eprint = "2603.02384",
    archivePrefix = "arXiv",
    primaryClass = "gr-qc",
    month = "3",
    year = "2026"
}

@article{Bixano:2026duf,
    author = "Bixano, Leonel and Matos, Tonatiuh",
    title = "{Potential Space Symmetries in Ernst-like Formulations of Einstein-Maxwell/ModMax-Scalar field Theories}",
    eprint = "2605.17843",
    archivePrefix = "arXiv",
    primaryClass = "gr-qc",
    month = "5",
    year = "2026"
}

@article{Bixano:2026uqk,
    author = "Bixano, Leonel and Matos, Tonatiuh",
    title = "{Exact rotating dilatonic branch in ModMax electrodynamics without Maxwell analogue}",
    eprint = "2604.13490",
    archivePrefix = "arXiv",
    primaryClass = "gr-qc",
    month = "4",
    year = "2026"
}

@article{Bixano:2026ouq,
  author = {Bixano, Leonel and Matos, Tonatiuh},
  title = {{Generalized Einstein-ModMax-ScalarField theories and new exact solutions}},
  eprint = {2603.26073},
  archiveprefix = {arXiv},
  primaryclass = {gr-qc},
  month = {3},
  year = {2026},
}

@article{Griffiths:2005qp,
    author = "Griffiths, J. B. and Podolsky, J.",
    title = "{A New look at the Plebanski-Demianski family of solutions}",
    eprint = "gr-qc/0511091",
    archivePrefix = "arXiv",
    doi = "10.1142/S0218271806007742",
    journal = "Int. J. Mod. Phys. D",
    volume = "15",
    pages = "335--370",
    year = "2006"
}

@article{Carter:1968ks,
    author = "Carter, B.",
    title = "{Hamilton-Jacobi and Schrodinger separable solutions of Einstein's equations}",
    doi = "10.1007/BF03399503",
    journal = "Commun. Math. Phys.",
    volume = "10",
    number = "4",
    pages = "280--310",
    year = "1968"
}

@article{Ovcharenko:2025cpm,
    author = "Ovcharenko, Hryhorii and Podolsk{\'y}, Ji{\v{r}}{\'\i}",
    title = "{New class of rotating charged black holes with nonaligned electromagnetic field}",
    eprint = "2508.04850",
    archivePrefix = "arXiv",
    primaryClass = "gr-qc",
    doi = "10.1103/8wkz-th6v",
    journal = "Phys. Rev. D",
    volume = "112",
    number = "6",
    pages = "064076",
    year = "2025"
}

@article{Kubiznak:2007kh,
    author = "Kubiznak, David and Krtous, Pavel",
    title = "{On conformal Killing-Yano tensors for Plebanski-Demianski family of solutions}",
    eprint = "0707.0409",
    archivePrefix = "arXiv",
    primaryClass = "gr-qc",
    doi = "10.1103/PhysRevD.76.084036",
    journal = "Phys. Rev. D",
    volume = "76",
    pages = "084036",
    year = "2007"
}

@article{Ovcharenko:2026pow,
    author = "Ovcharenko, Hryhorii and Podolsky, Jiri",
    title = "{Uniqueness of stationary axisymmetric type D black holes with non-aligned electromagnetic field}",
    eprint = "2604.13202",
    archivePrefix = "arXiv",
    primaryClass = "gr-qc",
    month = "4",
    year = "2026"
}

@article{Gray:2025lwy,
    author = "Gray, Finnian and Kubiznak, David and Ovcharenko, Hryhorii and Podolsky, Jiri",
    title = "{Hidden symmetries and separability structures of Ovcharenko-Podolsk{\'y} and conformal-to-Carter spacetimes}",
    eprint = "2511.21538",
    archivePrefix = "arXiv",
    primaryClass = "gr-qc",
    doi = "10.1103/8832-htpg",
    journal = "Phys. Rev. D",
    volume = "113",
    number = "4",
    pages = "044050",
    year = "2026"
}

@article{Matos:1994qv,
    author = "Matos, T. and Plebanski, J.",
    title = "{Axisymmetric stationary solutions as harmonic maps}",
    eprint = "gr-qc/9402044",
    archivePrefix = "arXiv",
    reportNumber = "CINVESTAV-12-93",
    doi = "10.1007/BF02108050",
    journal = "Gen. Rel. Grav.",
    volume = "26",
    pages = "477",
    year = "1994"
}

@article{Kocherlakota:2025cwq,
    author = "Kocherlakota, Prashant and Narayan, Ramesh",
    title = "{Doubly separable spacetimes and symmetry constraints on their self-gravitating matter content}",
    eprint = "2507.18706",
    archivePrefix = "arXiv",
    primaryClass = "gr-qc",
    doi = "10.1103/4wp1-2xny",
    journal = "Phys. Rev. D",
    volume = "112",
    number = "12",
    pages = "124058",
    year = "2025"
}

@article{Galtsov:2025nia,
    author = "Gal'tsov, Dmitri and Karsanov, Rostom",
    title = "{Gauged supergravities: Solutions with a Killing tensor}",
    eprint = "2503.06589",
    archivePrefix = "arXiv",
    primaryClass = "gr-qc",
    doi = "10.1103/PhysRevD.111.104011",
    journal = "Phys. Rev. D",
    volume = "111",
    number = "10",
    pages = "104011",
    year = "2025"
}

@article{Bicak:1999sa,
    author = "Bicak, J. and Pravda, Vojtech",
    title = "{Spinning C metric: Radiative space-time with accelerating, rotating black holes}",
    eprint = "gr-qc/9902075",
    archivePrefix = "arXiv",
    doi = "10.1103/PhysRevD.60.044004",
    journal = "Phys. Rev. D",
    volume = "60",
    pages = "044004",
    year = "1999"
}

@article{Astorino:2023elf,
    author = "Astorino, Marco and Boldi, Giovanni",
    title = "{Plebanski-Demianski goes NUTs (to remove the Misner string)}",
    eprint = "2305.03744",
    archivePrefix = "arXiv",
    primaryClass = "gr-qc",
    reportNumber = "IFUM-1105-FT, LIFT-3-1.23",
    doi = "10.1007/JHEP08(2023)085",
    journal = "JHEP",
    volume = "08",
    pages = "085",
    year = "2023"
}

@article{Anabalon:2012ta,
    author = "Anabalon, Andres",
    title = "{Exact Black Holes and Universality in the Backreaction of non-linear Sigma Models with a potential in (A)dS4}",
    eprint = "1204.2720",
    archivePrefix = "arXiv",
    primaryClass = "hep-th",
    doi = "10.1007/JHEP06(2012)127",
    journal = "JHEP",
    volume = "06",
    pages = "127",
    year = "2012"
}

@article{Nakajima:2026lkr,
    author = "Nakajima, Hiroaki and Guo, Ya and Lin, Wenbin",
    title = "{On uniqueness of non-vacuum stationary axisymmetric type D spacetime}",
    eprint = "2608.18653",
    archivePrefix = "arXiv",
    primaryClass = "gr-qc",
    month = "8",
    year = "2026"
}

@article{Furugori:2026qbb,
    author = "Furugori, Hideo and Tomizawa, Shinya",
    title = "{Generalized Black holes with Fully Non-aligned Electromagnetic Fields}",
    eprint = "2609.10332",
    archivePrefix = "arXiv",
    primaryClass = "gr-qc",
    reportNumber = "TTI-MATHPHYS-44",
    month = "9",
    year = "2026"
}

\end{document}